\documentclass[12pt]{article}

\usepackage[width=16cm, height=22cm]{geometry}   
\usepackage[font={small}, width=14cm]{caption}
\usepackage{amsmath}
\usepackage{amssymb}
\usepackage{amsthm}
\usepackage{amscd}
\usepackage{bbm}
\usepackage{enumitem}
\usepackage{float}
\usepackage{authblk}
\usepackage{mathrsfs}
\usepackage[mathscr]{eucal}
\usepackage[normalem]{ulem}
\usepackage{graphicx}

\newtheorem{thm}{Theorem}[section]
\newtheorem{prop}[thm]{Proposition}
\newtheorem{lem}[thm]{Lemma}
\newtheorem{cor}[thm]{Corollary}

\theoremstyle{definition}
\newtheorem{dfn}[thm]{Definition}

\theoremstyle{remark}
\newtheorem{rem}[thm]{Remark}

\newtheorem{example}[thm]{Example}

\newcommand{\CC}{\mathbb{C}}
\newcommand{\NN}{\mathbb{N}}
\newcommand{\RR}{\mathbb{R}}

\newcommand{\ZZ}{\mathbb{Z}}

\title{Delocalized spectra of Landau operators on helical surfaces}

\author[1]{Yosuke Kubota
\thanks{ykubota@shinshu-u.ac.jp}}
\author[2]{Matthias Ludewig\thanks{matthias.ludewig@mathematik.uni-regensburg.de}}
\author[3]{Guo Chuan Thiang\thanks{guochuanthiang@bicmr.pku.edu.cn}}
\affil[1]{\normalsize Department of Mathematical Sciences, Shinshu University}
\affil[1]{\normalsize RIKEN iTHEMS}
\affil[2]{\normalsize Fakult\"{a}t f\"{u}r Mathematik, Universit\"{a}t Regensburg}
\affil[3]{\normalsize Beijing International Center for Mathematical Research, Peking University}

\date{\today}

\begin{document}

\maketitle

\abstract{
On a flat surface, the Landau operator, or quantum Hall Hamiltonian, has spectrum a discrete set of infinitely degenerate Landau levels. We consider surfaces with asymptotically constant curvature away from a possibly non-compact submanifold, the helicoid being our main example. The Landau levels remain isolated, provided the spectrum is considered in an appropriate Hilbert module over the Roe algebra of the surface delocalized away from the submanifold. Delocalized coarse indices may then be assigned to them. As an application, we prove that Landau operators on helical surfaces have no spectral gaps above the lowest Landau level.
}

\subsection*{Introduction}

On the Euclidean plane, it is well known that the Laplacian twisted by a line bundle of nonzero constant curvature $b$ times the volume form (also called a \emph{Landau operator}), has spectrum $(2\NN+1)|b|$ comprising infinitely degenerate and evenly spaced isolated \emph{Landau levels}. This is known as Landau quantization in the physics literature \cite{Landau}. In the 1980s, with the discovery of the quantum Hall effect \cite{Klitzing}, it was realized that the Landau levels should be ``topological'' in some sense to account for the stability of the effect \cite{TKNN}, and that this has some manifestation on the boundary of the planar material as ``edge states'' \cite{Halperin} whose energies fill up the gaps between Landau levels. Mathematically, the ``topological protection'' of such spectral features of Landau operators should be the result of certain index theorems \cite{ASS, BES}.

So far, spectral analysis of Landau operators has mostly been limited to the \emph{constant} curvature situation, both in the curvature of the ${\rm U}(1)$ line bundle (physically the magnetic field strength perpendicular to the surface) and in the underlying surface geometry. Thus one studied subsets of the Euclidean plane, the product cylinder $S^1\times \RR$, the hyperbolic plane, etc., and applied constant magnetic fields. However, perfectly constant curvatures are merely idealizations, so the Landau level concept should be investigated in geometrically perturbed settings.

The stability of Landau levels (as essential spectra) against perturbations of a constant magnetic field vanishing at infinity were considered in \cite{Inahama, Iwatsuka}. A recent perspective \cite{LT, LTcobordism} is that the Euclidean/hyperbolic plane Landau level spectral projections, while infinitely degenerate and thus not characterized by a Fredholm index, actually realize the \emph{coarse} index \cite{Roebook} of an associated Dirac operator. This coarse geometric perspective, together with supersymmetry techniques, can account for small geometric perturbations very efficiently, see Section \ref{sec:indices.examples}.

\begin{figure}
	\centering
	\includegraphics[scale=0.32]{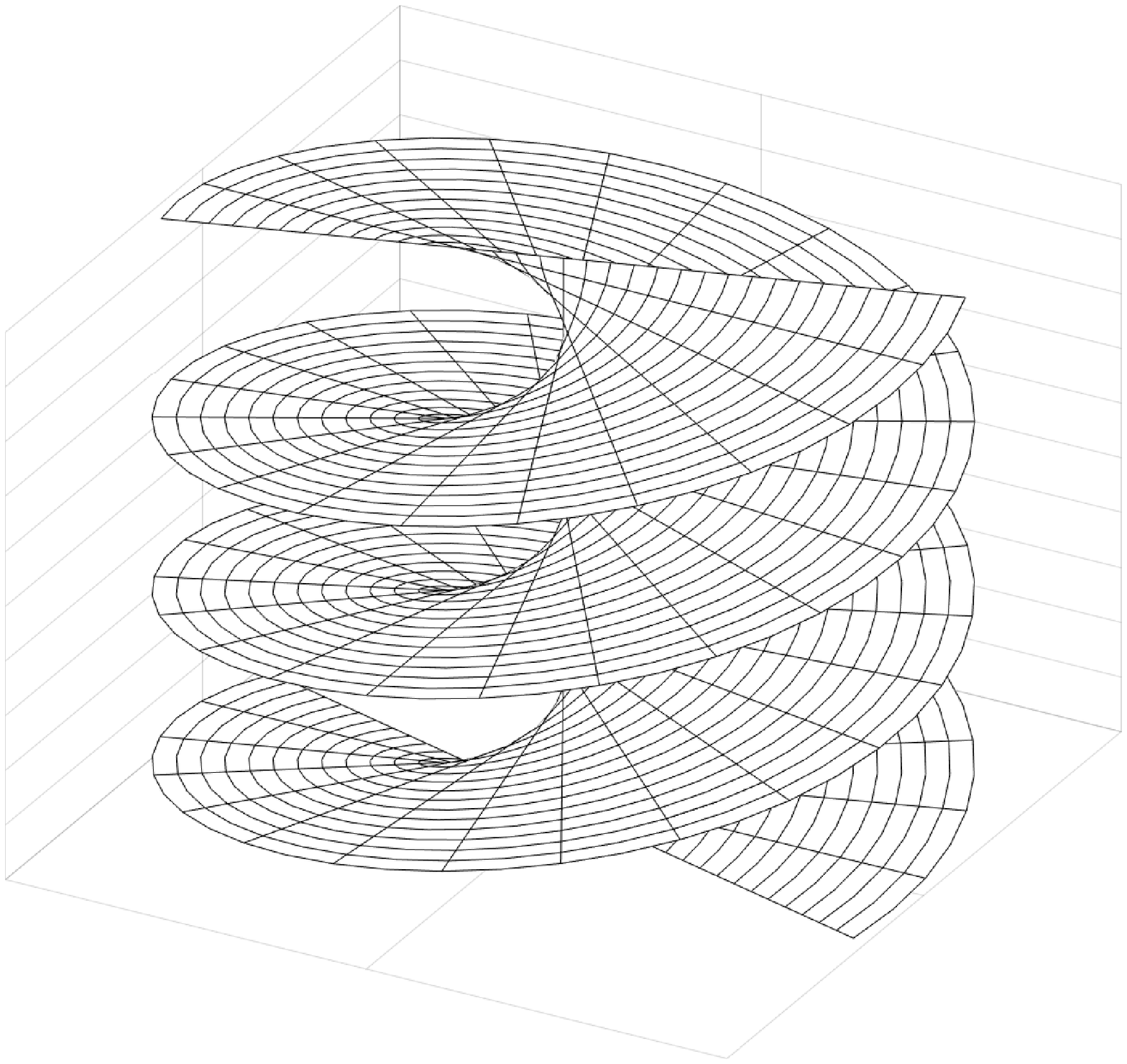}\includegraphics[scale=0.32]{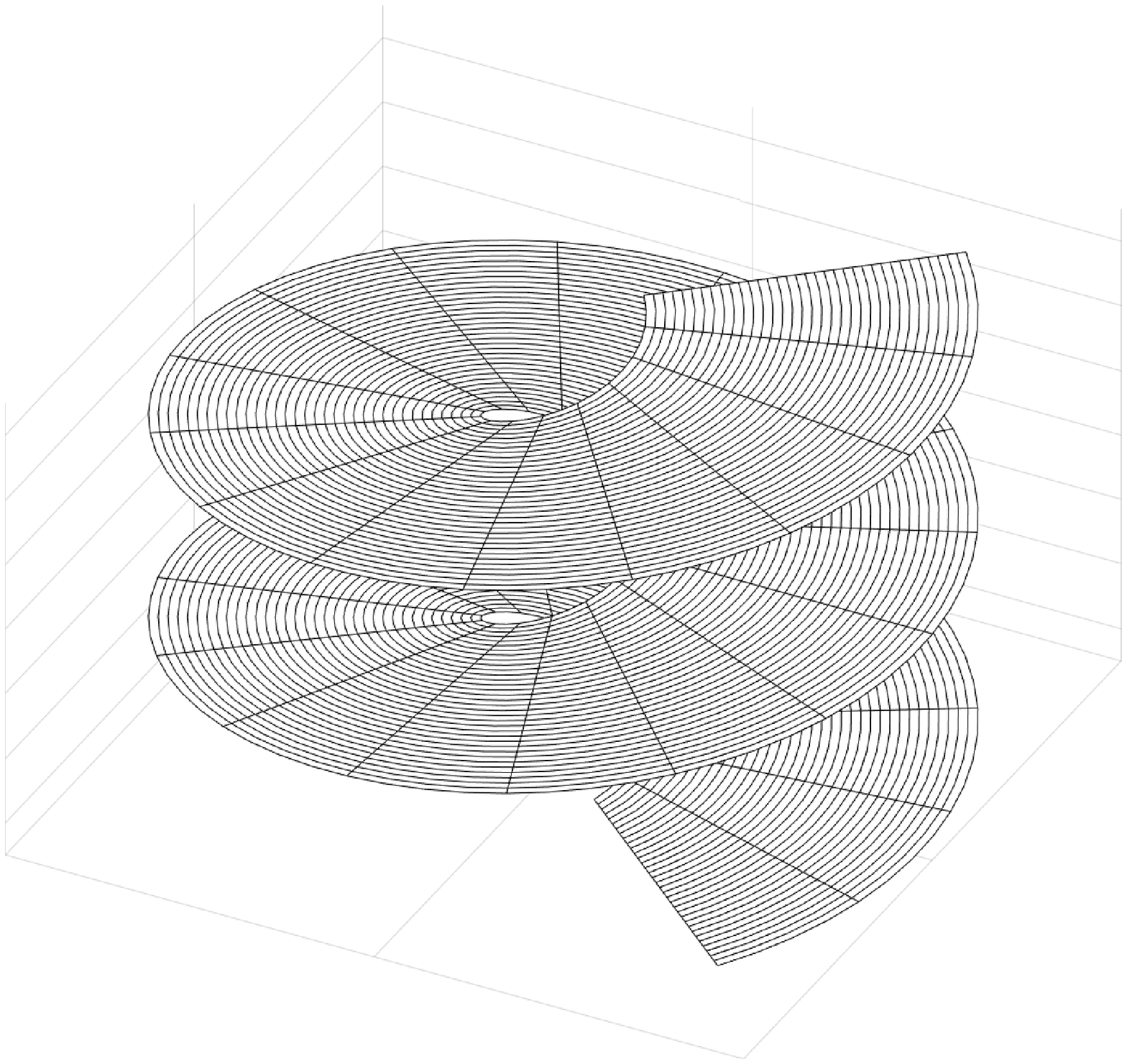} 
	\caption{[Left] A boundary-less helicoid is a minimal surface in $\RR^3$. [Right] A half-helicoid modelling a screw-dislocated surface. Its boundary is a helix winding around the dislocation axis.}\label{fig:helicoids}
\end{figure}

In physical practice, one has a surface embedded in Euclidean 3-space, and submits it to an externally applied magnetic field. An example of such a surface, which unfortunately does \emph{not} have asymptotically constant curvature, is the helicoid (see Sec.\ \ref{sec:helicoid.definition}). The (half-)helicoid models a screw-dislocated surface, as illustrated in Fig.\ \ref{fig:helicoids}. It is a physical heuristic, extrapolated from the Euclidean plane case, that \emph{helically} propagating states binding to the screw dislocation axis will appear in the spectrum of Landau operators on such screw-dislocated surfaces. This idea has been investigated in related \emph{discrete} models physically \cite{RZV}, and mathematically \cite{Kubota}, the latter with newly-developed coarse geometry methods. The relation between \cite{Kubota} and the present work on continuum Landau operators is discussed in Section \ref{sec:relation.discrete}. In fact, in the time-reversal symmetric setting, there have been experimental realizations of such helical states in acoustic topological insulators \cite{Xue,Ye}.

In a helicoid geometry, and also more generally, Landau levels are not spectrally isolated, even inside the essential spectrum. This presents difficulties in defining their indices via spectral projections in the usual way. To overcome this, we develop a perspective of Landau operators as operators on Hilbert modules over Roe $C^*$-algebras (Sect.\ \ref{sec:Dirac.on.Roe}). Using the spectral theory of such operators, we are able to show that the helicoid Landau levels are isolated in the part of the spectrum ``delocalized away from the screw dislocation axis'' (Theorem \ref{thm:quotient.Landau}). This allows for the construction of a \emph{delocalized index} for the Landau levels, \emph{without requiring them to be surrounded by a strict (essential) spectral gap}. Furthermore, we use a coarse Mayer--Vietoris principle to show that this index has a ``dimensional reduction'' to an index supported along the dislocation axis. The spectral implication is that a Landau operator on a helical surface (see Definition \ref{dfn:dislocated.Landau}) is forbidden from having any spectral gaps above the lowest Landau level, with helical edge states filling up the gaps between the delocalized Landau levels (Theorem \ref{thm:helical.gap.filling}). Due to the coarse geometry techniques, this result is very robust against geometric perturbations, see Remark \ref{rem:robust}. We also explain how our analysis specializes to a more concrete one involving Fourier transform and spectral flow, under the assumption of translation invariance along the dislocation axis (Sect.\ \ref{sec:Fourier}).

\section{Background and warm-up}\label{sec:warm-up}
Our main result of physical interest is Theorem \ref{thm:helical.gap.filling}, stating that the Landau operator on a helical surface has no spectral gaps above the lowest Landau level. Here, the helical surface is embedded in 3D Euclidean space with axis along the $z$-direction, and is subject to an external constant-strength magnetic vector field parallel to $z$.

Similar results in geometrically simple settings (e.g.\ Euclidean half-plane) can be obtained by direct spectral analysis. However, such methods are not (known to be) generalizable to other geometrical settings. Our method of proof is very general, and is of independent interest in spectral theory. It also leads to some new stability results for Landau levels (Section \ref{sec:indices.examples}). However, the mathematical techniques involved are relatively new and may be unfamiliar to the non-specialist. This warm-up section discusses the familiar models of the simple harmonic oscillator and Euclidean space Landau Hamiltonian, with a view towards our general coarse index theory perspective.

\subsection{Supersymmetric index of harmonic oscillator}\label{sec:SHO}
The simple harmonic oscillator Hamiltonian on the Euclidean line is
\begin{equation*}
H_{\rm SHO}=-\frac{d^2}{dx^2}+x^2,
\end{equation*}
where we have omitted the customary factor of $\frac{1}{2}$, and physical units. There is a first-order ``square root'',
\begin{align}
D_{\rm SHO}&=\begin{pmatrix}0 & -\frac{d}{dx}+x \\\frac{d}{dx}+x & 0\end{pmatrix}=:\begin{pmatrix} 0 & a^*\\ a &0\end{pmatrix},\label{eqn:supercharge}\\
\tilde{H}:=D_{\rm SHO}^2&=\begin{pmatrix} a^*a & 0 \\ 0 & aa^* \end{pmatrix}=\begin{pmatrix} H_{\rm SHO}-1 & 0 \\ 0 & H_{\rm SHO}+1 \end{pmatrix}=:\begin{pmatrix}\tilde{H}^B & 0 \\ 0 & \tilde{H}^F\end{pmatrix}.\nonumber
\end{align}
The operators $\tilde{H}^B=a^*a\geq 0$ and $\tilde{H}^F=aa^*\geq 0$ are, respectively, the \emph{bosonic} and \emph{fermionic} parts of the \emph{supersymmetric Hamiltonian} $\tilde{H}=\tilde{H}^B\oplus \tilde{H}^F$, and they share the same non-zero eigenvalues,
\begin{equation*}
\sigma(\tilde{H}^B)\setminus\{0\}=\sigma(\tilde{H}^F)\setminus\{0\}.
\end{equation*}
Since $\tilde{H}^F=\tilde{H}^B+2$, the spectrum $\sigma(\tilde{H}^B)$ is invariant under an energy shift by $2$, except for its 0 eigenvalue. It follows that $\sigma(H_{\rm SHO})=\sigma(\tilde{H}^B+1)=2\NN+1$. 

The operators $a,a^*$ are called \emph{ladder operators}, and they exchange the eigenspaces of $\tilde{H}^B$ and $\tilde{H}^F$ that have the same non-zero eigenvalues. In the context of supersymmetric quantum mechanics, the operator $D_{\rm SHO}$ is called a \emph{supercharge} for $\tilde{H}$. It is an odd operator in the sense that it maps the bosonic ($+1$-graded) subspace to the fermionic ($-1$-graded) subspace, and vice versa. The harmonic oscillator ground state space is thus identified with the ``supersymmetric ground state space'', ${\rm ker}(\tilde{H})$, of $\tilde{H}$. The latter kernel is purely ``bosonic'' (it is spanned by $\binom{e^{-x^2/2}}{0}$), thus the \emph{supersymmetric/Witten index} is $+1$. 

This ``algebraic'' discussion is well known in physics (e.g.\ \S10.2 of \cite{Hori}). Crucially, the index is stable under some deformations: replace Eq.\ \eqref{eqn:supercharge} by the odd operator
\begin{equation*}
D_h=\begin{pmatrix} 0 & -\frac{d}{dx} + h(x)\\
\frac{d}{dx} + h(x) & 0 \end{pmatrix},
\end{equation*}
with $h$ some smooth real-valued function, so the ground states of $D_h$ are given by $e^{-\int_0^x h}$ (bosonic) and/or $e^{\int_0^x h}$ (fermionic), depending on normalizability. In fact,  if $h$ is invertible outside a compact set, then $D_h$ is a \emph{Callias-Dirac} operator \cite{Callias}, whose index may be deduced from the large-scale behaviour of $h$ via an index formula, Eq.\ 3.2 of \cite{Callias}.

The harmonic oscillator example indicates that \emph{large-scale} data can lead to ``asymmetry'' (with respect to a $\ZZ_2$-grading) of the kernel of an odd Dirac-type Hamiltonian $D$. This asymmetry is manifested as the ground state space of the Laplace-type operator $\tilde{H}=D^2$. 

\subsection{Index of Landau Hamiltonian}\label{sec:Landau.index}
Let $b>0$. The Landau operator $H_b$ on the Euclidean plane is a ``pure geometric operator'' in the sense that it is simply the Laplace operator coupled to a connection on a ${\rm U}(1)$ line bundle with curvature $b$ times the Riemannian volume form. Physically, it is the Hamiltonian operator for a 2D electron gas subject to a uniform magnetic field perpendicular to the plane. Its explicit expression in Landau gauge is
\begin{equation*}
H_b=-\partial_x^2-(\partial_y-ibx)^2.
\end{equation*}
There is again a first-order ``square root'' (see Section \ref{sec:geometric.setup} for the geometric reason),
\begin{equation*}
D_b=\begin{pmatrix} 0 & -\partial_x+i(\partial_y-ibx) \\ \partial_x+i(\partial_y-ibx) & 0\end{pmatrix},\qquad D_b^2=\begin{pmatrix}H_b - b& 0 \\ 0 & H_b+b\end{pmatrix}.
\end{equation*}
The same argument as in the $H_{\rm SHO}$ case leads to the \emph{Landau quantization} result,
\begin{equation*}
\sigma(H_b)=(2\NN+1)b,
\end{equation*}
with the ground state space (lowest Landau level) being the ``purely bosonic'' Dirac kernel, ${\rm ker}(D_b)$. This kernel is easily determined by observing that $D_b$ is translation invariant in the $y$-direction. The Fourier transform is
\begin{equation*}
D_b(k_y)=\begin{pmatrix} 0 & -\frac{d}{dx}+bx-k_y\\ \frac{d}{dx}+bx-k_y & 0\end{pmatrix},\qquad k_y\in\widehat{\RR}.
\end{equation*}
A change of variable $\tilde{x}=bx-k_y$, turns $D_b(k_y)$ into $b\cdot D_{\rm SHO}$. So each $D_b(k_y)$ has a 1-dimensional ``bosonic'' kernel. Overall, ${\rm ker}(D_b)$ is infinite-dimensional, so how should we quantify the lowest Landau level of $H_b$ as an ``index of $D_b$''?

\paragraph{Equivariant index.} Consider the operators $T_x, T_y$ of translations by a distance $\sqrt{2\pi/b}$ in the $x$ and $y$ directions respectively. These lift to \emph{magnetic} translation operators $\tilde{T}_x=e^{i\sqrt{2\pi b}y}\cdot T_x$ and $\tilde{T}_y=T_y$ respectively, which commute with $D_b$ (and also $H_b$). The integral flux condition through a square of side length $\sqrt{2\pi/b}$ means that $\tilde{T}_x, \tilde{T}_y$ commute, so $D_b$ has an abelian lattice $\ZZ^2$ of symmetries. Fourier transform turns $D_b$ into a family $D_b(k_x,k_y)$ of operators parametrized by the \emph{magnetic Brillouin torus} $\widehat{\ZZ}^2$ dual to $\ZZ^2$. Each $D_b(k_x,k_y)$ is actually a Dirac operator on the square, subject to quasi-periodicity conditions labelled by $(k_x,k_y)\in \widehat{\ZZ}^2$. So the Fourier transform of $D_b$ is a family of elliptic operators on a compact manifold, and the kernel of $D_b$ is an eigen\emph{bundle} over $\widehat{\ZZ}^2$. More precisely, this eigenbundle is the \emph{index bundle} \cite{AS4} representing a $K$-theory class in $K^0(\widehat{\ZZ}^2)$. To a non-technical approximation, this $K$-theory class is the well-known TKNN--Chern number of a Landau level \cite{TKNN}. A difficulty arises when a periodic potential $V$ is added to $H_b$, with periodicity incommensurate with that of the magnetic translation lattice. This can be fixed by passing to a noncommutative Brillouin torus \cite{BES}.

\paragraph{Why is coarse index needed?} The $\ZZ^2$-equivariant approach to indices of Landau levels sketched above (commutative or otherwise) is rather deficient since we cannot go beyond $\ZZ^2$-invariant potentials or perturbations. It is also limited to the flat geometry and uniform magnetic field setting. It is therefore important to recognise that the $\ZZ^2$-equivariant approach is actually a special case of the \emph{coarse index}, Eq.\ \eqref{eqn:coarse.assembly}, which can be thought of as a way to quantify the supersymmetric ground state space of a Dirac operator (thus also a Landau level) on a general Riemannian manifold. Coarse indices, developed in \cite{Roebook,HR-book,HR-coarse,WillettYu}, take into account the large-scale (or coarse) geometric data while ignoring small-scale details.

\subsection{$K$-theory and (coarse) index}
A general reference on operator $K$-theory is \cite{Blackadar}. Let us recall some notions of index theory in this language.

A bounded \emph{Fredholm} operator $F\in\mathcal{B}(\mathscr{H})$ on a Hilbert space $\mathscr{H}$ has finite-dimensional kernel and cokernel, and its index is the difference of their dimensions. It is convenient to write
\begin{equation*}
{\rm Ind}(F):=\dim \ker F - \dim \ker F^* \equiv \dim \ker\begin{pmatrix} 0 & F^*\\ F& 0\end{pmatrix},
\end{equation*}
where the dimension on the right side is counted in the $\ZZ_2$-graded sense. 

A Fredholm operator may also be characterized by its invertibility modulo compact operators. In $C^*$-algebra language, we have the ideal of compact operators $\mathcal{K}(\mathscr{H})\subset\mathcal{B}(\mathscr{H})$, the quotient map $\pi:\mathcal{B}(\mathscr{H})\to \mathcal{B}(\mathscr{H})/\mathcal{K}(\mathscr{H})=:\mathcal{Q}(\mathscr{H})$ to the \emph{Calkin algebra}, and $\pi(F)$ becomes invertible in the Calkin algebra.

The invertible elements of $\mathcal{Q}(\mathscr{H})$ are organized into connected components labelled by an integer. This is exactly the $K_1$-theory group,
$
K_1(\mathcal{Q}(\mathscr{H}))\cong\ZZ.
$
This identification with the integers is given more precisely by a connecting map
\begin{equation}
\partial:K_1(\mathcal{Q}(\mathscr{H}))\overset{\cong}{\longrightarrow} \underbrace{K_0(\mathcal{K}(\mathscr{H}))}_{\ZZ}.\label{eqn:K.index}
\end{equation}
On the right side, the symbol $K_0(\cdot)$ denotes the homotopy classes of projections taken in a certain way. For instance, compact projection operators are finite-rank, and two such projections are homotopic (via a path of projections) iff they have the same rank. Classes in $K_0(\mathcal{K}(\mathscr{H}))$ are represented by a formal difference of finite-rank projections, $p_1\ominus p_2$, and the isomorphism $K_0(\mathcal{K}(\mathscr{H}))\cong \ZZ$ is given by $[p_1\ominus p_2]\mapsto \mathrm{rank}(p_1)-\mathrm{rank}(p_2)$. Under this identification, the map of Eq.\ \eqref{eqn:K.index} is precisely an index map (\S8.3.2 \cite{Blackadar}),
\begin{equation*}
\partial[\pi(F)]=[p_{\ker F}\ominus p_{\ker F^*}]\leftrightarrow {\rm Ind}(F).
\end{equation*}
In practice, one might have a Dirac operator on a compact manifold (thus with compact resolvent), and $F$ arises as its bounded transform (which is Fredholm),
\begin{equation*}
D=\begin{pmatrix} 0 & D_+^* \\ D_+ & 0\end{pmatrix}\rightsquigarrow \frac{D}{\sqrt{1+D^2}}=\begin{pmatrix}0 & F^* \\ F & 0\end{pmatrix},\quad F=D_+(1+D_+^*D_+)^{-1/2}.
\end{equation*}

\paragraph{Why Hilbert $C^*$-modules?}
The above operator $K$-theoretic formulation of the classical index admits a vast generalization. For example, we may have a family of Dirac operators on a compact manifold, parametrized by another compact mani\-fold $X$. Instead of a single Hilbert space $\mathscr{H}$, we have a bundle of Hilbert spaces, or more formally, a Hilbert $C(X)$-module on which the family acts. The spectral theory and functional calculus is similarly parametrized, and the families index defines an element of $K_0(C(X))\cong K^0(X)$, as represented by an index bundle over $X$, see \cite{AS4}. We saw an example of this in Section \ref{sec:Landau.index} --- the $\ZZ^2$-equivariant index of $D_b$ was analyzed by Fourier transforming to a Dirac family parametrized by $X=\widehat{\ZZ}^2$.

Even more generally, we can replace $C(X)$ by some other $C^*$-algebra $B$ suited to the problem at hand, and obtain $K_0(B)$-valued indices (``noncommutative families index''). For example, the coarse index of a Dirac operator on a noncompact $M$ is valued in $K_0(C^*(M))$, where $C^*(M)\subset \mathcal{B}(L^2(M))$ is the Roe $C^*$-algebra of $M$, whose definition is recalled in Section \ref{sec:Dirac.on.Roe}. When the Riemannian manifold $M$ is the Euclidean plane $\RR^2$, the lowest Landau level projection directly defines a (non-trivial) class in $K_0(C^*(\RR^2))$; similarly for $M$ the hyperbolic plane \cite{LT}. 

Motivated by this, we will generalize in Section \ref{sec:generalized.Fredholm} the ideas of supersymmetric quantum mechanics and supercharges to operators acting on Hilbert $C^*$-modules. In Section \ref{sec:Dirac.on.Roe}, we will explain a natural viewpoint of Dirac and Landau operators on $M$ as operators acting on a certain Hilbert $C^*(M)$-module. Much of the spectral theory of operators on Hilbert spaces generalizes to the Hilbert $C^*$-module setting \cite{Lance,Trout}. Although we recover the Hilbert space Dirac/Landau operators via the natural representation of $C^*(M)$ on $L^2(M)$ (Lemma \ref{lem:same.spec}), the Hilbert $C^*(M)$-module viewpoint is more powerful, as it allows us to do spectral theory in quotient algebras.

\paragraph{Why delocalized coarse index?}
The group $K_0(C^*(M))$ gives a way to count the size of infinite-rank projections in $C^*(M)\subset \mathcal{B}(L^2(M))$, such as the Landau level projections. It is insensitive to the addition or removal of finite-rank 
projections (Lemma \ref{rem:flat.LLL}), and we may work with the quotient algebra $C^*(M)/\mathcal{K}$ instead. This has the effect of ignoring the discrete part of the spectrum of operators, much like working in the Calkin algebra $\mathcal{Q}=\mathcal{B}/\mathcal{K}$ isolates the essential spectrum. 

We can go a step further. Euclidean space Landau levels are understood to comprise delocalized (or ``extended'') states, see \cite{Halperin,Kunz,BES}, and \cite{LTwannier2} for a Wannier basis perspective. Generally, introducing curvature causes the Landau levels to broaden into bands. This happens even if the curvature is only concentrated near some positive-codimension submanifold $N$, as occurs for a helicoid. If the bands overlap, then they do not define genuine spectral projections. However, if we work modulo the Roe algebra $C^*_M(N)$ localized near the high curvature region $N$, then the spectrum is reduced. Then it is possible to recover a spectral projection (and therefore $K$-theory index) in the quotient algebra $C^*(M)/C^*_M(N)$. This is a key idea for defining ``delocalized'' coarse indices for helical Landau operators, Eq.\ \eqref{eqn:delocalized.index} in Section \ref{sec:delocalized.index}. Heuristically, the non-triviality of a delocalized coarse index indicates the existence of spectrum delocalized away from $N$.

\section{Spectral supersymmetry}
In Section \ref{sec:warm-up}, we encountered the notion of ``spectral supersymmetry'' and supercharges (e.g.\ \S 5 of \cite{Thaller}, \cite{CKS}, \S10 of \cite{Hori}), and saw how it leads to an index-theoretic understanding of Euclidean space Landau levels. 
We shall abstract these notions to operators acting on Hilbert $C^*$-modules, motivated by the utility of viewing Dirac and Landau operators as acting on Hilbert $C^*$-modules over Roe algebras (Section \ref{sec:Dirac.on.Roe} later).

\subsection{Supercharges on Hilbert $C^*$-modules}
A standard reference on regular operators on Hilbert $C^*$-modules is \cite{Lance}, Section 9.

\medskip
Let $B$ be a $C^*$-algebra, and $\mathcal{E}_+,\mathcal{E}_-$ be countably generated Hilbert $B$-modules. A regular operator $D_+:\mathcal{E}_+\rightarrow \mathcal{E}_-$, is a closed, densely defined $B$-linear operator such that $D_-:=D_+^*:\mathcal{E}_-\rightarrow \mathcal{E}_+$ is densely defined (and closed), and $1+D_+^*D_+$ has dense range.

The $C^*$-algebra of adjointable ($B$-linear) bounded (resp.\ compact) operators $\mathcal{E}_+\rightarrow \mathcal{E}_-$ is denoted by $\mathcal{B}(\mathcal{E}_+, \mathcal{E}_-)$ (resp.\ $\mathcal{K}(\mathcal{E}_+,\mathcal{E}_-)$). 
When $\mathcal{E}_+=\mathcal{E}_-=\mathcal{E}$, we simply write $\mathcal{B}(\mathcal{E})$ and $\mathcal{K}(\mathcal{E})$.
A complex number $z$ is in the resolvent $\rho(T)$ of a regular operator $T$ on $\mathcal{E}$ if $(T-z)^{-1}$ exists in $\mathcal{B}(\mathcal{E})$, and is in the spectrum $\sigma(T)$ otherwise. We say that $T$ has compact resolvent if $(T -z)^{-1} \in \mathcal{K}(\mathcal{E})$ for any $z \in \rho(T)$.

For $B=\CC$, these reduce to the usual notions of resolvent and spectrum of a closed densely defined operator on a Hilbert space.

\begin{rem}[Interior tensor product]\label{rem:regular.homomorphism}
Let $A,B$ be $C^*$-algebras, $\mathcal{E}$ be a Hilbert $B$-module, $\mathcal{F}$ be a Hilbert $A$-module, and let $\varpi: B \to \mathcal{B}(\mathcal{F})$ be a $*$-representation by $A$-linear operators.
Then for a regular (self-adjoint) operator $D$ on $\mathcal{E}$, we can construct a regular (self-adjoint) operator $\varpi(D) = D \otimes _\varpi 1$ on $\mathcal{E} \otimes_{\varpi} \mathcal{F}$  (\cite{Lance}, Proposition 9.10), as the closure of the operator $\varpi(D)_0$ defined by
\[{\rm Dom}(\varpi(D)_{0}) = {\rm Dom}(D) \otimes_{\varpi , \mathrm{alg} } \mathcal{F}, \qquad  \varpi(D)_0 \cdot (x \otimes_\varpi y) = D \cdot x \otimes _\varpi y .
\]
Note that if $\varpi$ is injective, then $D$ and $\varpi(D)$ have the same spectrum. Indeed, the spectra of $D$ and $\varpi(D)$ are the same things as the Gelfand--Naimark spectra of the abelian $C^*$-algebras $\phi_D(C_0(\RR))$ and $\phi_{\varpi(D)}(C_0(\RR))$ respectively, where for a self-adjoint operator $T$ on a Hilbert module $\mathcal{E}$, $\phi_T : C_0(\RR) \to \mathcal{B}(\mathcal{E})$ denotes the continuous functional calculus of $T$ (see Theorem 10.9 of \cite{Lance}).
The injectivity of $\varpi$ implies that $\varpi \colon \phi_D(C_0(\RR)) \to \phi_{\varpi(D)}(C_0(\RR))$ is an isomorphism. 
\end{rem}

\begin{thm}\label{thm:ess.susy}
Let $D_+:\mathcal{E}_+\rightarrow\mathcal{E}_-$ be a regular operator. 
Then $\sigma(D_+^*D_+)\setminus\{0\}=\sigma(D_+D_+^*)\setminus\{0\}$. 
\end{thm}
\begin{proof}
When $B=\CC$, this is a standard fact in functional analysis (Lemma 2.1 of \cite{Inahama}, Theorem 2.9 of \cite{Moller}, Corollary 5.6 of \cite{Thaller}). For general $B$, let $(\pi,\mathscr{H})$ be a faithful $\ast$-representation of $B$. Following Remark \ref{rem:regular.homomorphism}, the interior tensor products 
\begin{equation*}
D_\pm \otimes_\pi 1_{\mathscr{H}} \colon \mathcal{E}_\pm \otimes_\pi \mathscr{H} \to \mathcal{E}_\mp \otimes_\pi \mathscr{H}
\end{equation*}
are closed operators between Hilbert spaces.
By definition,
\begin{align*}
(D_+ \otimes_\pi 1_{\mathscr{H}})^*(D_+ \otimes_\pi 1_{\mathscr{H}}) &= D_+^*D_+ \otimes_\pi 1_{\mathscr{H}}\\
(D_+ \otimes_\pi 1_{\mathscr{H}})( D_+\otimes_\pi 1_{\mathscr{H}})^*&=D_+D_+^* \otimes_\pi 1_{\mathscr{H}}
\end{align*}
hold. Moreover, we have $\sigma(D_+^*D_+ \otimes_\pi 1_{\mathscr{H}}) = \sigma(D_+^*D_+)$ and $\sigma(D_+D_+^* \otimes _\pi 1_{\mathscr{H}}) = \sigma(D_+D_+^*)$ by injectivity of $\pi$. 
This reduces the problem to the case of Hilbert space operators.
\end{proof}

Suppose $D$ is an odd self-adjoint regular operator on a graded Hilbert $B$-module $\hat{\mathcal{E}}=\mathcal{E}_+ \oplus \mathcal{E}_-$. Thus $D=\begin{pmatrix} 0 & D_- \\ D_+ & 0 \end{pmatrix}$ with $D_+ \colon \mathcal{E}_+ \to \mathcal{E}_-$ a regular operator and $D_-=D_+^*$.

\begin{dfn}\label{dfn:supersymmetry}
An \emph{(abstract) supercharge} on a graded Hilbert $B$-module $\hat{\mathcal{E}}$ is a triple $(D_\pm,H_\pm ,\theta_\pm)$, such that 
\begin{equation*}
D=\begin{pmatrix} 0 & D_-\\ D_+ & 0\end{pmatrix},\qquad D_-=D_+^*,
\end{equation*}
is an odd self-adjoint regular operator on $\hat{\mathcal{E}}$ with compact resolvent, 
\begin{equation*}
H = \begin{pmatrix}H_+ & 0 \\ 0 & H_- \end{pmatrix}
\end{equation*} 
is an even and positive regular operator on $\hat{\mathcal{E}}$, and $\theta_\pm$ are bounded self-adjoint operators on $\mathcal{E}_\pm$ such that 
\begin{equation}
     D^2=\begin{pmatrix}D_+^*D_+ & 0 \\ 0 & D_+D_+^*\end{pmatrix} = \begin{pmatrix}H_+ +\theta_+ & 0 \\ 0 & H_- + \theta_- \end{pmatrix}. \label{eqn:bootstrap_abst}
\end{equation}  
Furthermore, we say that the supercharge $(D_\pm,H_\pm,\theta_\pm)$ is \emph{flat} if $H_+=H_-$, and has \emph{constant shift parameters} if $\theta_\pm =b_\pm \cdot 1$ for $b_\pm \in \RR$. 
\end{dfn}

\subsection{Generalized Fredholm index of supercharges}\label{sec:generalized.Fredholm}
For a $C^*$-algebra $B$, let $\mathscr{H}_B:=\mathscr{H} \otimes_\CC B$ be the \emph{standard Hilbert $B$-module}, where $\mathscr{H}$ denotes the infinite dimensional separable Hilbert space.
There is a short exact sequence
\begin{equation}
0 \to \mathcal{K}(\mathscr{H}_B) \to \mathcal{B}(\mathscr{H}_B) \overset{\pi}{\to} \mathcal{B}(\mathscr{H}_B)/\mathcal{K}(\mathscr{H}_B) \to 0.\label{eqn:SES.Hilbert.module}
\end{equation}
The $C^*$-algebras $\mathcal{K}(\mathscr{H}_B)$ and $B$ are Morita equivalent with a preferred Morita equivalence given by $\mathscr{H}_B$, hence we have a canonical isomorphism $K_0(\mathcal{K}(\mathscr{H}_B)) \cong K_0(B)$.

\medskip

Let $\hat{\mathcal{E}} = \mathcal{E}_+ \oplus \mathcal{E}_-$ be a graded Hilbert $B$-module and let $(D_\pm, H_\pm, \theta_\pm)$ be an abstract supercharge on $\hat{\mathcal{E}}$.
Despite the symmetry of the nonzero spectrum expressed in Theorem \ref{thm:ess.susy}, an asymmetry can arise between the two graded components of the kernel of $D$, which can be measured by a generalized Fredholm index as follows.

By the Kasparov stabilization theorem \cite{Kasparov}, there are unitaries 
\begin{equation*}
U_\pm \colon \mathscr{H}_B \oplus \mathcal{E}_\pm \to \mathscr{H}_B.
\end{equation*}
By the assumption of $D$ having compact resolvent, the bounded transform 
\[F:=U_- ( 1_{\mathscr{H}_B} \oplus  D_+(1+D_+^*D_+)^{-1/2} )U_+^* \in \mathcal{B}(\mathscr{H}_B) \]
has $\pi(F) \in \mathcal{B}(\mathscr{H}_B)/\mathcal{K}(\mathscr{H}_B)$ being unitary, hence
\begin{equation}
{\rm Ind}(D):=\partial [\pi(F)] \in K_0(\mathcal{K}(\mathscr{H}_B)) \cong K_0(B) \label{eqn:coarse.index.boundary}
\end{equation}
is defined, with $\partial:K_1(\mathcal{B}(\mathscr{H}_B)/\mathcal{K}(\mathscr{H}_B))\to K_0(\mathcal{K}(\mathscr{H}_B))$ the boundary map for the extension \eqref{eqn:SES.Hilbert.module}. This ${\rm Ind}(D)$ is defined independently of the choice of $U_\pm$ at the $K$-theory level.  

Furthermore, if zero is isolated in the spectrum of $D$, the kernel of $D^2$ (thus also of $D$) is obtainable by continuous functional calculus.  This kernel is then a projective submodule of $\mathcal{E}$ which splits into a positively graded component and a negatively graded component, and we may write (\cite{HR-book}, Proposition 4.8.10 (c))
\begin{equation} \label{DefinitionCoarseIndex}
  \mathrm{Ind}(D) = \left[U_+ P_0(D_{+}^*D_{+})U_+^*\right] - \left[U_-P_0(D_{+}D_{+}^*)U_-^*\right] \in K_0(\mathcal{K}(\mathscr{H}_B))\cong K_0(B),
\end{equation}
where $P_0(\cdot )$ denotes the kernel projection. 
When $\lambda\in\RR$ is an isolated eigenvalue, we also write $P_\lambda(\cdot)$ for the $\lambda$-eigenprojection.

Observe that if an abstract supercharge has strictly positive shift parameters $\theta_+$ and $\theta_-$, then $D$ is invertible and its index necessarily vanishes. On the other hand, if only one of $\theta_+, \theta_-$ is strictly positive while the other is a non-positive constant, we have the following.

\begin{prop}\label{prop:LLL} 
Let $(D_\pm,H_\pm, \theta_\pm)$ be an abstract supercharge over $\hat{\mathcal{E}}$, such that ${\rm Ind}(D) \neq 0$. Assume that $\theta _+ = b_+ \cdot 1$ and $\theta _- \geq b_- \cdot 1$ for some constants $b_+ \leq 0$  and $b_->0$ (resp.\ $\theta_+ \geq b_+ \cdot 1$ and $\theta_-=b_- \cdot 1$ for some $b_+ >0$ and $b_-  \leq 0 $). 
Then the bottom of $\sigma(H_+)$ (resp.\ $\sigma(H_-)$) is an isolated point $-b_+$ (resp.\ $-b_-$) called the \emph{lowest Landau level (LLL)}, and 
\begin{equation*}
[P_{-b_+}(H_+)]=\mathrm{Ind}(D) \qquad \text{resp.}\quad [P_{-b_-}(H_-)]=-\mathrm{Ind}(D)
\end{equation*}
 in $K_0(B)$.
\end{prop}

\begin{proof} Consider the first case, $\theta _+ = b_+ \cdot 1\leq 0$ and $\theta _- \geq b_- \cdot 1>0$.
On the right hand side of Eq.\ \eqref{eqn:bootstrap_abst}, the bottom-left piece, $H_-+\theta_-=H_-+b_-$, is strictly positive with the interval $(0,b_-)$ a spectral gap. Thus
\begin{equation*}
P_0(D_+D_+^*) = P_0(H_- + b_-) = 0.
\end{equation*}
By Theorem \ref{thm:ess.susy}, the top-left piece in Eq.\ \eqref{eqn:bootstrap_abst}, $H_+ + b_+$, also has $(0,b_-)$ as a spectral gap. It is furthermore a non-negative operator as it is (a piece of) the square of the operator $D$.
Thus the kernel of $H_+ + b_+$ must be spectrally isolated, and  
\begin{equation*}
P_0(D_+^*D_+)=P_0(H_++b_+) = P_{-b_+}(H_+).
\end{equation*}
Consequently we get
\begin{align*}
[U_+P_{-b_+}(H_+)U_+^*]&=[U_+P_0(D_+^*D_+)U_+^*]\\
&=[U_+P_0(D_+^*D_+)U_+^*]-[U_-P_0(D_+D_+^*)U_-^*]\\
& \overset{\mathrm{Eq.}\,\eqref{DefinitionCoarseIndex}}{=}{\rm Ind}(D) \in K_0(B),
\end{align*}
which is non-zero by assumption.

Similarly, for $\theta_+ \geq b_+ \cdot 1>0$ and $\theta _- = b_-\leq 0$, the top-left piece in Eq.\ \eqref{eqn:bootstrap_abst} is strictly positive with gap $(0,b_+)$. Then for the bottom-right piece, we deduce that the kernel is isolated, with \emph{negatively}-graded spectral projection $P_0(H_-+b_-)=P_{-b_-}(H_-)$ representing $-{\rm Ind}(D)$. 
\end{proof}

If a compact interval $I\subset \RR$ is such that $I\cap\sigma(\cdot)$ is separated in the spectrum of a self-adjoint regular operator, the corresponding spectral projection, denoted $P_I(\cdot)$, is a regular projection obtainable by continuous functional calculus.

\begin{prop}\label{prop:MvN}
Let $(D_\pm,H_\pm,\theta_\pm)$ be an abstract supercharge over $\hat{\mathcal{E}}=\mathcal{E}\oplus \mathcal{E}$. Let $I\subset(0,\infty)$ be a compact interval such that $I\cap\sigma(D_+^*D_+)$ is isolated in $\sigma(D_+^*D_+)$. Then $U_+P_I(D_+^*D_+)U_+^*$ and $U_-P_I(D_+D_+^*)U_-^*$ define the same class in $K_0(\mathcal{K}(\mathscr{H}_B)) \cong K_0(B)$.
\end{prop}

\begin{proof}
By Theorem \ref{thm:ess.susy}, $I\cap\sigma(D_+D_+^*)$ is also isolated in $\sigma(D_+D_+^*)$. Define the operators \mbox{$p:=U_+P_I(D_+^*D_+)U_+^*$}, $q:=U_-P_I(D_+D_+^*)U_-^*$, and 
\[ T:=U_-D_+P_I(D_+^*D_+)U_+^*  = U_-P_I(D_+D_+^*)D_+U_+^*.\]
Then, by the assumption of $D$ having compact resolvent, $p$, $q$ are projections in $\mathcal{K}(\mathscr{H}_B)$.
Moreover, since $D_+ \colon p\mathscr{H}_B \to q\mathscr{H}_B$ is bounded and invertible, we have $T=Tp=qT$, and that $T^*T+(1-p)$, $TT^*+(1-q)$ are invertible on $\mathscr{H}_B$. 
Hence
 \[ v:=T (T^*T + (1-p) )^{-1/2} = (TT^*+ (1-q))^{-1/2}T \]
 is a partial isometry implementing the Murray--von Neumann equivalence between $v^*v = T^*(TT^*+(1-q))^{-1}T = p$ and $vv^*= T(T^*T + (1-p))^{-1}T^*=q $. 
\end{proof}

\medskip

The following ``bootstrap'' result generalizes the ladder operator method discussed in Sections \ref{sec:SHO}--\ref{sec:Landau.index}, and is motivated by the idea of \emph{shape invariance} in supersymmetric quantum mechanics \cite{CKS}.
\begin{prop}\label{prop:LL.iteration}
Let $(D_{n,\pm} ,H_{n,\pm} ,\theta_{n,\pm}), n\in\NN$, be a sequence of abstract supercharges over $\hat{\mathcal{E}}_n$ with constant shift parameters $\theta_{n,\pm} = b_\pm \cdot 1$. 
Assume that 
\begin{enumerate}
    \item $\mathcal{E}_{n+1,+} = \mathcal{E}_{n,-}$ and $H_{n+1,+} = H_{n,-}$ hold for all $n \in\NN$,\label{assum.1}
    \item $b_->0$, and $\Delta:=b_- - b_+ >0$.\label{assum.2}
\end{enumerate}
Then for each $n\in\NN$, the spectrum of $H_{n,\pm}$ is contained in the discrete set $\Delta \cdot \NN - b_+$ of abstract Landau levels. The $m$-th Landau level $\Delta\cdot m - b_+$ is attained in $\sigma(H_{n,+})$ iff $-b_+$ is attained in $\sigma(H_{n+m,+})$.

Suppose further, that 
\begin{enumerate}[resume]
\item ${\rm Ind }(D_n) \neq 0$ for all $n \in \NN$.\label{assum.3}
\end{enumerate}
Then for each $n\in\NN$, we have $\sigma(H_{n,\pm})=\Delta\cdot\NN-b_+$, and the $K_0(B)$ class of the spectral projection for $\Delta\cdot m-b_+$ equals ${\rm Ind}(D_{n+m})$ for every $m\in\NN$.
\end{prop}
\noindent
\begin{example}
If $(D_{\pm}, H_\pm=H, b_\pm)$ is a flat abstract supercharge with constant shift parameters $b_\pm$, where $b_->0$ and $\Delta:=b_- - b_+ >0$, then the sequence $(D_{n,\pm},H_{n,\pm}, \theta_{n,\pm})=(D_\pm,H,b_\pm)$ satisfies the assumptions \ref{assum.1}--\ref{assum.2} of Prop.~\ref{prop:MvN}.
\end{example}
\begin{proof}
From the spectral supersymmetry, Theorem \ref{thm:ess.susy}, we have the following equality of subsets of $[0,\infty)$,
\begin{equation}
\sigma(H_{n,+}+b_+)\setminus\{0\}=\sigma(H_{n,-} + b_-)=\sigma(H_{n+1,+} + b_+ )+ \Delta .\label{eqn:abstract.bootstrap}
\end{equation}
Thus $H_{n,+}+b_+$ has $(0,\Delta)$ as a spectral gap. 
By the same argument applied to $H_{n+1,\pm }$, we also have that $H_{n+1,+}+b_+$ has $(0,\Delta)$ as a spectral gap. This, together with Eq.\ \eqref{eqn:abstract.bootstrap}, implies that $H_{n,+} + b_+$ has $(\Delta, 2\Delta)$ as a spectral gap. 
By an inductive argument, we deduce that $\sigma(H_{n,+}+b_+)\subset \Delta \NN$, with $\Delta\cdot m $ attained iff $0$ is attained in $\sigma (H_{n+m,+}+b_+)$. 

With the assumption $0\neq {\rm Ind} (D_n)$, zero cannot be missing in the spectrum, and we conclude that $\sigma(H_{n,+} + b_+)=\Delta \NN$. By Prop.\ \ref{prop:MvN}, for $m\geq 1$, the $\Delta\cdot m$ spectral projections for $H_{n,+} + b_+$ and $H_{n,-} + b_-$ are Murray-von Neumann equivalent. 
Since $H_{n,-}=H_{n+1,+}$, this can be reformulated as the statement that the $\Delta \cdot m - b_+$ spectral projection for $H_{n,+}$ is equivalent to the $\Delta \cdot (m-1) - b_+$ spectral projection for $H_{n+1,+}$, and also to the $-b_+$ spectral projection of $H_{n+m,+}$ by iterating $m$ times. Note that the non-trivial index implies $b_+\leq 0$, and Prop.\ \ref{prop:LLL} applies to give $\mathrm{Ind}(D_{n+m})=[U_+P_{-b_+}(H_{n+m,+})U_+^*]=[U_+P_{\Delta\cdot m-b_+}(H_{n,+})U_+^*]\in K_0(B)$.
\end{proof}

\begin{rem}\label{rem:LL.iteration.opposite.sign}
In the case where $b_+ >0$ and $b_+ - b_- >0$, the same proof shows that $\sigma (H_{n,\pm}) = \Delta \cdot \NN -b_- $, and the $\Delta \cdot m -b_-$ spectral projection of $H_{n,-}$ is equivalent to $-{\rm Ind}(D_{n+m})$ in $K_0(B)$.
\end{rem}

\section{Landau and Dirac operators on curved surfaces}
\label{sec:geometric.setup}

Let $M$ be a complete noncompact Riemannian spin surface, with metric tensor $g$, volume form $\omega$, and scalar curvature function $R$. Note that $H^2(M)=0$ so complex line bundles over $M$ are trivializable.
The spin Dirac operator on $M$ is an odd operator on the $\ZZ_2$-graded spinor bundle $S=S^+\oplus S^-$. We may twist $S$ by a (trivializable) Hermitian line bundle $\mathcal{L}$ with connection $A$, whose curvature may be written as $dA=\Theta\cdot\omega$ for some scalar function $\Theta$.

Let $C_b^\infty(M)$ denote the (real-valued) bounded smooth functions on $M$, and
\begin{equation}
C_\flat^\infty(M):=\{V\in C_b^\infty(M)\,:\, \|dV\|_\infty, \|d^*dV\|_\infty < \infty\}.\label{eqn:bounded.derivatives}
\end{equation}
Here, $d$ is the exterior derivative, $d^*$ is the codifferential, and $\|\cdot\|_\infty$ refers to the supremum norm.
 \emph{Throughout this paper, we will assume that $R,\Theta\in C_\flat^\infty(M)$.}

Write $D_\Theta = D_\Theta^{(M)}$ for the Dirac operator on $M$ twisted by a line bundle with connection $A$. The \emph{magnetic Laplacian} is 
\[
H_\Theta = H_\Theta^{(M)}=(d-iA)^*(d-iA),
\] also called the \emph{Landau operator} in physics. It describes the motion of an electron on $M$ subject to a magnetic field of strength $\Theta$. We will often drop the superscript $(M)$ when there is no confusion about the manifold in question.

The operators $H_\Theta, D_\Theta$ are essentially self-adjoint on the smooth compactly supported sections (\cite{Shubin}, \cite{HR-book} \S 10.2), and we use the same symbols for their closures to self-adjoint operators.  They are related by the Lichnerowicz--Weitzenb\"{o}ck formula, which may be written as (Prop.\ 2.1 of \cite{LT})
\begin{equation}
 D_{\Theta}^2=\begin{pmatrix}H_{\Theta - \frac{R}{4}}-\Theta+\frac{R}{4} & 0 \\ 0 & H_{\Theta+ \frac{R}{4}}+\Theta +\frac{R}{4}\end{pmatrix}\geq 0.\label{eqn:bootstrap.symmetric}
\end{equation}
Equivalent forms of \eqref{eqn:bootstrap.symmetric} are
\begin{align}
 D_{\Theta+\frac{R}{4}}^2 &=\begin{pmatrix}H_\Theta-\Theta & 0 \\ 0 & H_{\Theta+\frac{R}{2}}+\Theta+\frac{R}{2}\end{pmatrix}\geq 0,\label{eqn:bootstrap1}\\
D_{\Theta-\frac{R}{4}}^2 &=\begin{pmatrix}H_{\Theta-\frac{R}{2}}-\Theta+\frac{R}{2} & 0 \\ 0 & H_{\Theta}+\Theta\end{pmatrix}\geq 0.\label{eqn:bootstrap2}
\end{align}
In these formulae, we have implicitly used an identification of the trivializable line bundles on which the operators are acting, but the choice does not matter up to gauge equivalence.

\begin{rem}
If $H^1(M)=0$, then all choices of connection $A$ with the same curvature $\Theta\cdot\omega$ are gauge equivalent, so there is no ambiguity in writing $D_\Theta, H_\Theta$. Otherwise, there may be a moduli space of gauge inequivalent $A$ on $\mathcal{L}$ with the same curvature (corresponding to addition of ``Aharanov--Bohm fluxes''). This ambiguity occurs, and is accounted for, in Section \ref{sec:catenoid}.
\end{rem}

\subsection{Landau and Dirac Hamiltonians as operators on Hilbert $C^*$-modules over Roe algebras}\label{sec:Dirac.on.Roe}

Equations \eqref{eqn:bootstrap1} and \eqref{eqn:bootstrap2} are concrete versions of the supercharge relation defined in Eq.\ \eqref{eqn:bootstrap_abst}. When $R\equiv 0$ and $\Theta=b\in\RR\setminus\{0\}$, this expresses the well-known observation that the twisted Dirac operator on the Euclidean plane is a supercharge for (two shifted copies of) the Landau operator. Then the ``ladder operator trick'' (cf.\ Prop.\ \ref{prop:LL.iteration}) immediately shows that Landau operator's spectrum is $(2\NN+1)|b|$, see Section \ref{sec:Landau.index}.

For general $M$, it is natural to think of the Dirac and Landau operators as acting on certain Hilbert $C^*$-modules over the Roe $C^*$-algebra $C^*(M)$, in order to understand the Landau levels through generalized Fredholm indices (cf.\ \cite{LT} for $M=\RR^2$ and $M=\mathbb{H}$). Here, $C^*(M)$ is defined to be the $C^*$-algebra closure of the $\ast$-algebra $\mathbb{C}[M]$ of finite-propagation, locally compact operators on $L^2(M)$. The $C_0(\RR)$-functional calculi of $D_\Theta$ and $H_\Theta$ land within $\mathrm{M}_2(C^*(M))$ and $C^*(M)$ respectively (Prop. 3.6 of \cite{Roebook}), and the relevant Hilbert $C^*(M)$-modules are constructed as follows.

\medskip

Let $\phi_{D_\Theta } \colon C_0 (\RR ) \to \mathrm{M}_2(C^*(M)) = \mathcal{B}(C^*(M)^{\oplus 2})$ be the $\ast$-homomorphism defined by $\phi_{D_\Theta}(\varphi):=\varphi(D_\Theta)$. This is an odd homomorphism when $C_0(\RR)$ is given the even--odd grading, and $\mathrm{M}_2(C^*(M))$ is given the diagonal--off-diagonal grading. 

Following Trout \cite{Trout} \S 3, we define the unbounded operator $\mathcal{D}_\Theta$ as the interior tensor product $x \otimes_{\phi_{D_\Theta }}1$ in the sense of Remark \ref{rem:regular.homomorphism}, where $x$ denotes the identity function on $\RR$. More explicitly, the operator $\mathcal{D}_\Theta$ acts on the Hilbert $C^*(M)$-submodule 
\begin{align} 
\hat{\mathcal{E}}_\Theta := C_0(\RR) \otimes_{\phi_{D_\Theta}} C^*(M)^{\oplus 2} = \overline{\phi_{D_\Theta}(C_0(\RR)) \cdot C^*(M)^{\oplus 2}} \subset C^*(M)^{\oplus 2}, \label{eq:domain.functional.calculus}
\end{align}
and a core for this operator is $\phi_{D_\Theta}(C_c(\RR)) \cdot C^*(M)^{\oplus 2}$, on which it acts as  
\begin{align*}
    \mathcal{D}_\Theta (\phi_{D_\Theta}(\varphi) \cdot T):= \phi_{D_\Theta}(x\varphi) \cdot T.
\end{align*} 
This is closable, and its closure $\mathcal{D}_\Theta$ is regular and self-adjoint.

Similarly, the Landau operator $H_\Theta$ defines a 
$\ast$-homomorphism $\phi_{H_\Theta} \colon C_0(\RR_{\geq 0}) \to C^*(M)$ via $\phi_{H_\Theta} (\varphi):=\varphi(H_\Theta)$. On the Hilbert $C^*(M)$-submodule 
\begin{align}
    \mathcal{E}_\Theta := \overline{\phi_{H_\Theta}(C_0(\RR_{\geq 0})) \cdot C^*(M)} \subset C^*(M),\label{eqn:Landau.Roe.module1}
\end{align}  
we have the regular positive operator $\mathcal{H}_\Theta $ defined by 
\begin{align}
    \mathcal{H}_\Theta (\phi_{H_\Theta}(\varphi) \cdot T):= \phi_{H_\Theta}(x\varphi) \cdot T \label{eqn:Landau.Roe.module2}
\end{align}
on the core $\phi_{H_\Theta}(C_c(\RR_{\geq 0})) \cdot C^*(M) \subset {\rm Dom} (\mathcal{H}_\Theta)$.

\begin{rem}
It may be shown that $C^*(M)$ is not $\sigma$-unital when $M$ is a non-compact manifold (Exercise 5.4.8(ii) of \cite{WillettYu}). So by Prop.\ 12.3.1 of \cite{Blackadar}, the Hilbert module $\mathcal{E}_\Theta$ is a proper submodule of $C^*(M)$. Similarly, $\hat{\mathcal{E}}_\Theta\neq C^*(M)^{\oplus 2}$.
\end{rem}

\begin{rem}
As is stated in \cite{Trout}, Theorem 3.2, the operators $\varphi(\mathcal{D}_\Theta), \varphi\in C_0(\RR)$, are all compact operators on $\hat{\mathcal{E}}_\Theta$. This is seen from
\begin{align}
    \varphi(\mathcal{D}_\Theta) = \varphi (x) \otimes _{\phi_{D_\Theta}} 1 \in \mathcal{K}(C_0(\RR) \otimes_{\phi_{D_\Theta}} C^*(M)^{\oplus 2}). \label{eqn:interior.tensor}
\end{align}
Indeed, for a $\ast$-homomorphism $\varpi \colon B \to \mathcal{K}(\mathcal{F})$ of $C^*$-algebras (where $\mathcal{F}$ is a Hilbert $A$-module) and a Hilbert $B$-module $\mathcal{E}$, we have $T \otimes_{\varpi } 1  \in \mathcal{K}(\mathcal{E} \otimes _\varpi \mathcal{F})$ for any $T \in \mathcal{K}(\mathcal{E})$ (\cite{Lance}, Proposition 4.7). 
\end{rem}

\begin{lem}\label{lem:domain.module}
For $\Theta, R \in C_\flat^\infty(M)$, we have $\hat{\mathcal{E}}_{\Theta} = \mathcal{E}_{\Theta - \frac{R}{4}} \oplus \mathcal{E}_{\Theta + \frac{R}{4}}$.
\end{lem}
\begin{proof}
The Hilbert module $\hat{\mathcal{E}}_\Theta$ is defined through $D_\Theta$ in Eq.\ \eqref{eq:domain.functional.calculus}, while the Hilbert module $\mathcal{E}_\Theta$ is defined through $H_\Theta$.
We have to relate these definitions.
To this end, set $\tilde{\mathcal{E}}_\Theta := \mathcal{E}_{\Theta - \frac{R}{4}} \oplus \mathcal{E}_{\Theta + \frac{R}{4}}$ and
\[ 
 \tilde{H}_\Theta := \begin{pmatrix} H_{\Theta - \frac{R}{4}} & 0 \\
 0 & H_{\Theta + \frac{R}{4}} \end{pmatrix}, 
 \qquad 
 \tilde{\theta} := \begin{pmatrix} -\Theta + \frac{R}{4} & 0 \\ 0 & \Theta + \frac{R}{4} \end{pmatrix}.
 \]
 We have to show that $\tilde{\mathcal{E}}_\Theta = \hat{\mathcal{E}}_\Theta$.
Note that $D_\Theta^2=\tilde{H}_\Theta+\tilde{\theta}$ by Eq.\ \eqref{eqn:bootstrap.symmetric}. 
Since the functions $(x^2+1)^{-1} \in C_0(\RR)$ and $(x+1)^{-1} \in C_0(\RR_{\geq 0})$ are strictly positive,
the spaces $(x^2+1)^{-1} \cdot C_0(\RR) \subset C_0(\RR)$ and $(x+1)^{-1} \cdot C_0(\RR_{\geq 0}) \subset C_0(\RR_{\geq 0})$ are dense. 
(Indeed, $C_c(\RR) \subset C_0(\RR)$ is dense, and for $f \in C_c(\RR)$, we have $f = (x^2+1)^{-1} \cdot (f \cdot (x^2+1))$; the argument for $C_0(\RR_{\geq 0})$ is similar.)
Therefore, the submodules 
\begin{align*}
    \phi_{D_\Theta}((x^2+1)^{-1}) \cdot C^*(M)^{\oplus 2} \subset \hat{\mathcal{E}}_\Theta, \qquad \phi_{\tilde{H}_\Theta}((x+1)^{-1}) \cdot C^*(M)^{\oplus 2} \subset \widetilde{\mathcal{E}}_{\Theta}
\end{align*}
are dense. 
The operator $(\tilde{H}_\Theta+1)(D_\Theta^2+1)^{-1}= 1 - \tilde{\theta}\cdot\phi_{D_\Theta}((x^2+1)^{-1})$ is bounded, with the inverse $(D_\Theta^2+1)(\tilde{H}_\Theta+1)^{-1}$ being bounded as well, hence 
\begin{align*}
    \phi_{D_\Theta}((x^2+1)^{-1}) = (D_\Theta^2+1)^{-1} &=  (\tilde{H}_\Theta+1)^{-1}(\tilde{H}_\Theta +1)(D_\Theta^2+1)^{-1}\\
    &= \phi_{\tilde{H}_\Theta}((x+1)^{-1}) \cdot (1- \tilde{\theta} \cdot \phi_{D_\Theta}((x^2+1)^{-1}))
\end{align*}
implies that
\[ \hat{\mathcal{E}}_\Theta  = \overline{\phi_{D_\Theta}((x^2+1)^{-1}) \cdot C^*(M)^{\oplus 2}} = \overline{\phi_{\tilde{H}_\Theta}((x+1)^{-1})C^*(M)^{\oplus 2}} =\widetilde{\mathcal{E}}_\Theta. \qedhere \]
\end{proof}

The proof of Lemma \ref{lem:domain.module} also shows that $\mathcal{D}_\Theta^2$ and 
\begin{equation*}
\tilde{\mathcal{H}}_\Theta=\begin{pmatrix}\mathcal{H}_{\Theta-\frac{R}{4}} & 0 \\ 0 & \mathcal{H}_{\Theta+\frac{R}{4}}\end{pmatrix}
\end{equation*}
 have the common core 
\[ \hat{\mathcal{E}}_\Theta ^0:=\phi_{D_\Theta}((x^2+1)^{-1}) \cdot C^*(M)^{\oplus 2} = \phi_{\tilde{H}_\Theta}((x+1)^{-1}) \cdot C^*(M)^{\oplus 2},\]
and hence the difference 
\begin{equation} \label{defVarThetaDiff}
\tilde{\vartheta}\equiv\begin{pmatrix} \vartheta_+ & 0 \\ 0 & \vartheta_- \end{pmatrix}:= \mathcal{D}_\Theta^2 - \begin{pmatrix} \mathcal{H}_{\Theta - \frac{R}{4}}  & 0 \\ 0 & \mathcal{H}_{\Theta + \frac{R}{4}} \end{pmatrix}  
\end{equation}
is a densely-defined symmetric operator on $\hat{\mathcal{E}}_\Theta$. 
We will relate this mismatch $\vartheta_\pm$ with the curvature operators $\mp \Theta + \frac{R}{4}$ in Lemma \ref{lem:same.spec}.

\begin{lem}\label{lem:function.module}
Let $V,\Theta \in C_\flat^\infty(M)$ be real-valued. Left multiplication with $V$ acting on $C^*(M)$ preserves the submodule $\mathcal{E}_\Theta$, and acts as a bounded operator on it. 
\end{lem}
\begin{proof}
For any $\xi \in C^*(M)$, we have 
\begin{align}
    V \phi_{H_\Theta}((x+1)^{-1})\xi &= V \cdot (H_\Theta +1)^{-1} \cdot  \xi \nonumber \\
    &= [V, (H_\Theta +1)^{-1}] \xi + (H_\Theta +1)^{-1} V \xi \nonumber\\
    &= - (H_\Theta +1)^{-1}[V, H_\Theta] (H_\Theta +1)^{-1} \xi  + (H_\Theta +1)^{-1} V  \xi \nonumber\\
    &= (H_\Theta +1)^{-1} \big( -[V, H_\Theta  ] (H_\Theta +1)^{-1} +V \big) \xi .\label{eqn:membership.in.module}
\end{align} 
It is clear that $V \xi \in C^*(M)$.
We will show that $[V, H_\Theta  ](H_\Theta +1)^{-1}  \xi \in C^*(M)$, which implies that the right hand side of Eq.\ \eqref{eqn:membership.in.module} is in $\mathcal{E}_\Theta$. 
Indeed, by using the equalities $\alpha^* \circ \beta = \langle \alpha, \beta \rangle = \beta^* \circ \alpha $ for any compactly supported $1$-forms $\alpha, \beta \in \Omega^1_c(M)$, $[d^*, V] = - ([d,V])^* = -(dV)^*$, and
\[ d^*d(fg) = d^*d(f) \cdot g -2 \langle df, dg \rangle + f \cdot d^*d(g)\]
for any $f,g \in C_c^\infty(M)$, we get $[d^*d , V] = d^*d(V) - 2 (dV)^* \circ d$. Hence $[H_\Theta,V]$ is calculated as
\begin{align*}
    [H_\Theta, V] &= [(d -iA)^*(d-iA) , V] \\
     &= [d^*d, V] - i([d^*, V] \circ A  - A^* \circ  [d,V]) \\
    &=d^*d(V) -2 (dV)^* \circ d   -i(-(dV)^*\circ A - A^*\circ(dV)) \\
    &=d^*d(V) - 2 (dV)^* \circ (d - iA).
\end{align*}
By assumption, $V\in C_\flat^\infty(M)$ (see Eq.\ \eqref{eqn:bounded.derivatives}), so $||dV||<\infty, ||d^*dV||<\infty$. Also, the operator $(d - iA) (H_\Theta + 1)^{-1}$ is bounded, so $[H_\Theta, V](H_\Theta +1)^{-1}$ is well-defined as a bounded operator. 
Moreover, upon identifying 1-forms with functions in a trivialization, the terms $d^*d(V)$ and $(dV)^*$ are bounded multiplication operators and $(H_\Theta+1)^{-1},(d-iA)(H_\Theta+1)^{-1} \in C^*(M)$, thus $[H_\Theta, V](H_\Theta + 1)^{-1}$ is also in $C^*(M)$. 
\end{proof}

\medskip

We apply the interior tensor product (Remark \ref{rem:regular.homomorphism}) to $B=C^*(M)$, $\mathcal{F} = L^2(M)$, and the standard $\ast$-representation
\begin{equation*}
\pi \colon C^*(M) \to \mathcal{B}(L^2(M)).
\end{equation*} 
Then $\pi(\mathcal{D}_\Theta)$, $\pi(\tilde{\mathcal{H}}_\Theta)$ and $\pi(\tilde{\vartheta}_\pm)$ are self-adjoint operators on
\[ \hat{\mathcal{E}}_\Theta \otimes _\pi L^2(M) = C_0(\RR) \otimes _{\phi_{D_\Theta}} C^*(M)^{\oplus 2} \otimes _\pi L^2(M) \cong L^2(M)^{\oplus 2}, \]
where the isomorphism is given by $\varphi \otimes _{\phi_{D_\Theta}} K \otimes_\pi \xi \mapsto \varphi(D_\Theta) \cdot K \cdot \xi$.

\begin{lem}\label{lem:same.spec}
We have $\pi(\mathcal{D}_\Theta)=D_\Theta$ and $\pi(\mathcal{H}_\Theta) = H_\Theta$. Therefore, the spectrum of $\mathcal{H}_\Theta$ and $\mathcal{D}_\Theta$ as regular operators on $\mathcal{E}_\Theta$ and $\hat{\mathcal{E}}_\Theta$, coincide with that of $H_\Theta$ and $D_\Theta$ respectively. Moreover, the differences $\vartheta_\pm$ defined in Eq.\ \eqref{defVarThetaDiff} coincide with $\mp \Theta + \frac{R}{4}$, acting on $\mathcal{E}_{\Theta \pm \frac{R}{4}}$ as in Lemma \ref{lem:function.module}.
\end{lem}

\begin{proof}
As proved in Proposition 3.1 of \cite{Trout}, for $\varphi\in C_0(\RR)$ the functional calculus $\varphi(\mathcal{D}_\Theta) \in \mathcal{K}(\hat{\mathcal{E}}_\Theta)$ coincides with the restriction of $\phi_{D_\Theta}(\varphi)$ to the $C^*$-submodule $\hat{\mathcal{E}}_\Theta$ of $C^*(M)^{\oplus 2}$ (cf.\ Eq.\ \eqref{eqn:interior.tensor}). 
Hence we have $\varphi(\mathcal{D}_\Theta) \otimes_\pi 1_{L^2(M)}  = \phi_{D_\Theta}(\varphi)$ on the dense subspace $\hat{\mathcal{E}}_\Theta \otimes_{\rm alg, \pi} L^2(M)$ of $L^2(M)^{\oplus 2}$, thus
\[ \varphi(\mathcal{D}_\Theta \otimes_\pi 1) = \varphi(\mathcal{D}_\Theta) \otimes _\pi 1_{L^2(M)} = \phi_{D_\Theta}(\varphi)\]
for any $\varphi \in C_0(\RR)$. 
Since a self-adjoint operator can be reconstructed from its $C_0$-functional calculus, we get $\mathcal{D}_\Theta \otimes_\pi 1 = D_\Theta$.

The same argument also shows $\pi(\mathcal{H}_\Theta) = H_\Theta$. Finally, the function $\mp \Theta +\frac{R}{4}$ acting as bounded operators on $\mathcal{E}_{\Theta\mp\frac{R}{4}}$ as in Lemma \ref{lem:function.module}, satisfies $\pi(\mp \Theta +\frac{R}{4}) = \mp \Theta +\frac{R}{4}$ (where the right hand side is a multiplication operator on $L^2(M)$), and hence
\begin{align*}
     \pi (\mathcal{D}_\Theta^2) - \pi\begin{pmatrix} -\Theta + \frac{R}{4} & 0 \\ 0 & \Theta + \frac{R}{4} \end{pmatrix} 
&= D_\Theta^2 - \begin{pmatrix} -\Theta + \frac{R}{4} & 0 \\ 0 & \Theta + \frac{R}{4} \end{pmatrix} \\
&= \tilde{H}_\Theta = \pi \begin{pmatrix} \mathcal{H}_{\Theta - \frac{R}{4}}& 0 \\ 0 & \mathcal{H}_{\Theta +\frac{R}{4}} \end{pmatrix}. 
\end{align*} 
By injectivity of $\pi$, we get $\vartheta_\pm = \mp \Theta + \frac{R}{4}$.  

\end{proof}

Replacing $\Theta$ by $\Theta+\frac{R}{4}$ in Lemmas \ref{lem:domain.module} and \ref{lem:same.spec}, we arrive at the following relation between Dirac and Landau operators viewed as Hilbert $C^*(M)$-module operators.

\begin{cor}\label{cor:ASS.surface}
Let $D_\pm:= \left(\mathcal{D}_{\Theta + \frac{R}{4}}\right)_\pm$, $H_+:=\mathcal{H}_\Theta$, $H_-:= \mathcal{H}_{\Theta + \frac{R}{2}}$, $\theta _+:=-\Theta$ and $\theta_-:=\Theta + \frac{R}{2}$. Then the data $(D_\pm, H_\pm, \theta_\pm)$ determine an abstract supercharge over $\hat{\mathcal{E}}_{\Theta + \frac{R}{4}}$ (Definition \ref{dfn:supersymmetry}).
\end{cor}

\subsection{Supercharges over delocalized Roe algebra}
Let $N$ be a submanifold of $M$. 
Then the localized Roe algebra $C^*_M(N)$ is the $C^*$-algebra closure of the set of finite-propagation, locally compact operators on $L^2(M)$ that are supported on some  finite $r$-neighborhood of $N$. It is a closed ideal of $C^*(M)$. Let 
\begin{equation*}
\varpi \colon C^*(M) \to C^*(M)/C^*_M(N) 
\end{equation*}
denote the quotient $\ast$-homomorphism onto the \emph{delocalized Roe algebra}, and recall the construction in Remark \ref{rem:regular.homomorphism}.

\begin{lem}\label{lem:relative_coarse}
Let $\mathcal{H}_\Theta$ be the Landau operator on the Hilbert $C^*(M)$-module $\mathcal{E}_\Theta$, as defined in Eq. \eqref{eqn:Landau.Roe.module1} and \eqref{eqn:Landau.Roe.module2}. Let $V\in C_\flat^\infty(M)$ be a real-valued function on $M$ such that $V(x) \to 0$ as ${\rm dist}(x,N) \to \infty$. 
Then $\mathcal{H}_\Theta + V$, where $V$ acts on $\mathcal{E}_\Theta$ as in Lemma \ref{lem:function.module}, is also a self-adjoint regular operator on $\mathcal{E}_\Theta $, and the equality
\[ \varpi(\mathcal{H}_\Theta + V)  = \varpi(\mathcal{H}_\Theta ) \]
holds as regular operators on $\mathcal{E}_\Theta \otimes _\varpi C^*(M)/C^*_M(N)$. 
\end{lem}
\begin{proof}
Since $V \in C_\flat^\infty(M)$, by Lemma \ref{lem:function.module}, the sum $\mathcal{H}_\Theta + V$ is well-defined, and is again a self-adjoint operator. 
By Remark \ref{rem:regular.homomorphism}, $\varpi(\mathcal{H}_\Theta)$ and $\varpi(\mathcal{H}_\Theta +V)$ share the core 
\[ {\rm Dom}(\mathcal{H}_\Theta+V) \otimes_{\varpi, {\rm alg}}  C^*(M)/C^*_M(N)= {\rm Dom}(\mathcal{H}_\Theta) \otimes_{\varpi, {\rm alg}} C^*(M)/C^*_M(N).\] 
By definition of $\varpi(\mathcal{H}_\Theta+V)$, for any $T \in {\rm Dom}(\mathcal{H}_\Theta)$ and $X \in C^*(M)/C^*_M(N)$, we have
\begin{align*}
\varpi(\mathcal{H}_\Theta +V)(T \otimes_\varpi X) &= \varpi(\mathcal{H}_\Theta T + V\cdot T) \cdot X \\
&= \varpi( \mathcal{H}_\Theta T)X + \varpi(V\cdot T)X \\
&= \varpi(\mathcal{H}_\Theta )(T \otimes_\varpi X),
\end{align*}
where $\varpi(V\cdot T)=0$ comes from $V \cdot T \in C^*_M(N)$. This finishes the proof. 
\end{proof}

\paragraph{Coarse Dirac index ${\rm Ind}(D)$.}
The spin Dirac operator on $M$ can be viewed as an odd self-adjoint regular operator $D$ on a Hilbert submodule of $C^*(M)^{\oplus 2}$ with compact resolvent, and it has an index ${\rm Ind}(D)\in K_0(C^*(M))$ in the sense of Eq.\ \eqref{eqn:coarse.index.boundary}. This is the same thing as the \emph{coarse index} \cite{Roebook}.
Specifically, there is a \emph{coarse assembly map} \cite{HR-coarse},
\begin{equation}
\mu_M:K_0(M)\rightarrow K_0(C^*(M)),\label{eqn:coarse.assembly}
\end{equation}
which when applied to the $K$-homology class of the Dirac operator gives its coarse index $\mathrm{Ind}(D)$.

\begin{thm}\label{thm:ASS.relative}
Let $M$ be a complete Riemannian spin surface and let $N \subset M$ be a submanifold such that any $r$-neighborhood of $N$ is a proper subset of $M$. 
Assume that the scalar curvature $R$ and the magnetic field $\Theta = b + \Theta_{\rm pert}$ satisfy $R(x) \to 0$, $\Theta_{\rm pert}(x)  \to 0$ as ${\rm dist}(x,N)\to \infty$, where $0\neq b \in \RR$ is a constant. 
We define 
\[
\hat{\mathcal{E}}_n:=\hat{\mathcal{E}}_{\Theta + \frac{(2n+1)R}{4}} \otimes_\varpi C^*(M)/C^*_M(N)
\]
 and
\begin{align*}
    D_{n} &= \varpi \big(\mathcal{D}_{\Theta + \frac{(2n+1)R}{4}}\big), \\
     H_{n,+} &= \varpi\big(\mathcal{H}_{\Theta + \frac{nR}{2}}\big), \\
      H_{n,-} &= \varpi\big( \mathcal{H}_{\Theta + \frac{(n+1)R}{2}}\big), \\
        \theta_{n,\pm} &= \mp b\cdot 1.
\end{align*}
Then $(D_{n,\pm}, H_{n,\pm}, \theta_{n,\pm})$ defines a sequence of abstract supercharges over the Hilbert \mbox{$C^*(M)/C^*_M(N)$-module} $\hat{\mathcal{E}}_n$, satisfying assumptions \ref{assum.1}--\ref{assum.2} of Prop.\ \ref{prop:LL.iteration}. If furthermore $\varpi_*\mathrm{Ind}(D)\neq 0$, then assumption \ref{assum.3} of Prop.\ \ref{prop:LL.iteration} is satisfied as well.
\end{thm}
\begin{proof}
First, suppose $b>0$.
The basic relation Eq.\ \eqref{eqn:bootstrap1} in this perturbed setting is
\begin{equation*}
\mathcal{D}_{\Theta+\frac{(2n+1)R}{4}}^2=
\begin{pmatrix} \mathcal{H}_{\Theta+\frac{nR}{2}}-(b+\Theta_{\rm pert}+\frac{nR}{2}) & 0 \\ 0 & \mathcal{H}_{\Theta+\frac{(n+1)R}{2}}+(b+\Theta_{\rm pert}+\frac{(n+1)R}{2})\end{pmatrix}.
\end{equation*}
By Lemma \ref{lem:relative_coarse}, 
\begin{align*}
D_n^2\equiv\varpi \Big( \mathcal{D}_{\Theta+\frac{(2n+1)R}{4}}\Big)^2 &=\begin{pmatrix} \varpi \big( \mathcal{H}_{\Theta + \frac{nR}{2}} \Big) -b & 0 \\ 0 & \varpi\big( \mathcal{H}_{\Theta+\frac{(n+1)R}{2}}\big)+b\end{pmatrix}\\
&\equiv\begin{pmatrix}H_{n,+}-b & 0 \\ 0 & H_{n,-}+b\end{pmatrix}.
\end{align*}
This shows that $(D_{n,\pm},H_{n,\pm},\mp b)$ is a sequence of abstract supercharges over $\hat{\mathcal{E}}_{n}$. Furthermore, the assumptions \ref{assum.1}--\ref{assum.2} of Prop.\ \ref{prop:LL.iteration} are satisfied. As for assumption \ref{assum.3}, the coarse index ${\rm Ind}(D_{\Theta})$ does not depend on $\Theta$, hence ${\rm Ind}(D_n) = {\rm Ind}(D)$ for all $n\in\NN$, and similarly for the image under $\varpi_*$.

For $b<0$, we use the basic relation Eq.\ \eqref{eqn:bootstrap2} instead, to get a sequence of abstract supercharges satisfying the assumptions of Prop.\ \ref{prop:LL.iteration} with the sign change described in Remark \ref{rem:LL.iteration.opposite.sign}.
\end{proof}

\begin{rem}
Even if $R$ vanishes as ${\rm dist}(x,N) \to \infty$, there is no guarantee that the ``quotient supercharge'' becomes flat, that is, $H_{0,+}=\varpi(\mathcal{H}_{\Theta }) \neq \varpi(\mathcal{H}_{\Theta+\frac{R}{2}})=H_{0,-}$. This is because changing a field strength from $\Theta$ to $\Theta+\frac{R}{2}$ requires changing the connection $1$-form $A$, and so $H_{\Theta+\frac{R}{2}}-H_\Theta$ is not simply a $C_\flat^\infty(M)$ function. The best we can do is to construct a sequence of quotient supercharges as in Theorem \ref{thm:ASS.relative}.
\end{rem}

\subsection{Indices of Landau levels}\label{sec:indices.examples}
\subsubsection{Constant field: isolated lowest Landau level}
\begin{prop}\label{prop:LLL.concrete} 
Let $M$ be a noncompact complete Riemannian spin surface, such that ${\rm Ind}(D) \neq 0$ holds. Then, for constant magnetic field strength $\Theta=b\in\RR\setminus\{0\}$, with  
\begin{equation}
|b|>-\frac{1}{2} \inf_{x \in M} R(x),\label{eqn:LLL.field.constraint}
\end{equation}
the bottom of $\sigma(H_b)$ is an isolated point $|b|$, the lowest Landau level (LLL). Furthermore, the class of the spectral projection for the LLL is nontrivial in $K_0(C^*(M))$, with a sign change when $b$ changes sign.
\end{prop}
\begin{proof}
Set $\kappa := \inf_{x \in M} R(x)$. First suppose $b=|b|>-\frac{\kappa }{2}$. Corollary \ref{cor:ASS.surface} says that $(D_\pm,H_\pm,\theta_\pm)$ with $D_\pm=\left(\mathcal{D}_{b+\frac{R}{4}}\right)_\pm$, $H_+=\mathcal{H}_b$, $H_-=\mathcal{H}_{b+\frac{R}{2}}$, $\theta_+ =-b $ and $\theta_- = b + \frac{R}{2} >0$ is an abstract supercharge over $\hat{\mathcal{E}}_{b+\frac{R}{4}}$. Hence Prop.\ \ref{prop:LLL} applies to give $b$ as the LLL of $H_+$, which is also that of $H_b$ by Lemma \ref{lem:same.spec}. The $b<0$ case is similar.
\end{proof}

\begin{example}
When $M$ is the Euclidean or hyperbolic plane, the coarse Baum--Connes conjecture is verified, so that $\mathrm{Ind}(D)$ is a generator of $K_0(C^*(M))\cong\ZZ$. Prop.\ \ref{prop:LLL} says that the LLL spectral projection realizes this index class, as was found in \cite{LT}.
\end{example}

\begin{rem}\label{rem:flat.LLL}
For connected $M$, if the LLL spectral projection is non-trivial in $K_0(C^*(M))$, it must be infinitely degenerate (and is said to be a \emph{flat band} of essential spectrum), due to Lemma \ref{lem:finite.projection.trivial} below. Furthermore, this flatness is independent of the scalar curvature $R$, beyond the constraint \eqref{eqn:LLL.field.constraint} needed to spectrally isolate $|b|$. 
\end{rem}

\begin{lem}\label{lem:finite.projection.trivial}
Let $M$ be a connected, noncompact, complete Riemannian manifold.
If $P\in C^*(M)$ is a finite-rank projection, then $[P]=0$ in $K_0(C^*(M))$. 
\end{lem}
\begin{proof}
Without loss of generality, suppose $P=|\psi\rangle\langle\psi|$ is a rank-1 projection onto the span of $\psi\in L^2(M)$. 
Let $(\psi_n)_{n \in \NN}$ be a sequence of compactly supported functions such that $\psi_n \to \psi$ in $L^2(M)$. 
Then $P_n = |\psi_n\rangle\langle\psi_n|$ has finite propagation and is compact, hence contained in $C^*(M)$. 
Moreover, we have $P_n \to P$ in operator norm, hence also $P \in C^*(M)$.

Now, let $K \subset M$ be any compact subset. 
Then each $P_n$ has support near $K$ (i.e., support within a ball of finite radius around $K$, not necessarily uniformly in $n$).
Therefore each $P_n$ and consequently also $P$ is contained in the localized Roe algebra at $K$, $C^*_M(K) \subset C^*(M)$.
Hence the $K$-theory class $[P] \in K_0(C^*(M))$ is contained in the image of the map $K_0(C^*_M(K)) \to K_0(C^*(M))$.

However, this map is the zero map: As $M$ is geodesically complete, for each point $x \in M$, there exists a \emph{half-line} starting at $x$, i.e., a geodesic $\gamma : [0, \infty) \to M$ with $\gamma(0) = x$ such that the Riemannian distance $d(\gamma(s), \gamma(t)) = |t-s|$ for all $t, s \geq 0$.
Consequently, $\gamma$ provides an isometric embedding of $[0, \infty)$ into $M$; let $L$ denote its image. 
Therefore, the inclusion map $C^*_M(K) \to C^*(M)$ factors through $C^*_M(K \cup L)$, whose $K$-theory is isomorphic to that of $C^*(K \cup L)$.
But $K$ is compact, so $K\cup L$ and the half-line $L$ have Roe algebras with isomorphic $K$-theory, which is known to be trivial for the latter (\cite{HRY} \S 7 Prop.\ 1).
Consequently, the map $K_0(C^*_M(K)) \to K_0(C^*(M))$ factors through zero.
\end{proof}

\begin{rem}\label{rem:localized.is.compact}
The localized Roe algebra $C^*_M(K)$ is the ideal of compact operators $\mathcal{K}$, see pp.\ 22 of \cite{Roebook}. It may be shown, by a more abstract argument, that $K_0(C^*_M(K))=K_0(\mathcal{K}) \to K_0(C^*(M))$ is the zero map, see pp.\ 23 of \cite{Roebook}.
\end{rem}

\subsubsection{Constant field and flat surface: isolated higher Landau levels}
\begin{prop}\label{prop:Euclidean.LL}
Let $M$ be a noncompact complete Riemannian spin surface. 
 Suppose $M$ has vanishing scalar curvature, $R=0$, i.e., $M$ is either $S^1 \times \RR$, or $\RR^2$. Then once $b\neq 0$, the Landau operator $H_b^{(M)}$ has spectrum the Landau levels $(2\NN+1)|b|$.
Furthermore, the class of the $(2n+1)|b|$ spectral projection in $K_0(C^*(M))$ is independent of $n\in \NN$.
\end{prop}
\begin{proof}
First, suppose $b>0$. Let $(D_\pm,H_\pm,\theta_\pm)=((\mathcal{D}_b)_\pm,\mathcal{H}_b,\mp b)$ be the abstract supercharge given in Corollary \ref{cor:ASS.surface}. The sequence $(D_{n,\pm}, H_{n,\pm}, \theta_{n,\pm}):=(D_\pm,H_\pm,\mp b)$ satisfies the assumptions \ref{assum.1}--\ref{assum.2} of Prop.\ \ref{prop:LL.iteration}.
Thus for all $n\in\NN$, $\sigma(H_{n,+})\subset(2b)\cdot\NN+b=(2\NN+1)b$. Furthermore, for each $m\in\NN$, $(2m+1)b$ is attained in $\sigma(H_{n,+})=\sigma(H_+)$ iff $b$ is attained in $\sigma(H_{n+m,+})=\sigma(H_+)$. Together with Lemma \ref{lem:same.spec}, this shows that the non-empty set $\sigma(H_b)=\sigma(\mathcal{H}_b)=\sigma(H_+)$ must be the entire set $(2\NN+1)b$. Prop.\ \ref{prop:LL.iteration} also shows that each eigenprojection represents the same class $\mathrm{Ind}(D)$ in $K_0(C^*(M))$. The $b<0$ case is similar (see Remark \ref{rem:LL.iteration.opposite.sign}).
\end{proof}

\begin{example}\label{ex:Euclidean.cylinder}
For $M$ the Euclidean plane $\RR^2$, $(2\NN+1)|b|$ are the famous \emph{Landau levels} of $H_b^{(\RR^2)}$. Each Landau level spectral projection represents a generator of $K_0(C^*(\RR^2))\cong\ZZ$.
\end{example}
\begin{example}
For $M$ the cylinder $S^1\times\RR$ with product metric, we may modify the connection $A$ by a flat connection with holonomy $e^{ik}$ around $S^1$, but the resulting Landau operator spectrum remains $(2\NN+1)|b|$, independent of such choices. Although these Landau levels are ``trivial'' due to $K_0(C^*(S^1\times\RR))=0$, we can think of $H_b^{(\RR^2)}$ as a lift to the universal cover, and the indices in $K_0(C^*(\RR^2))$ as the non-trivial ``higher'' indices for the Landau levels of $H_b^{(S^1\times\RR)}$.
\end{example}

\begin{rem}
In \cite{LT}, a similar result was obtained for $M$ the hyperbolic plane, which has constant negative curvature. In this case, only finitely many Landau levels are isolated, with the precise number dependent on the size of $|b|$ (see \cite{Comtet}).
\end{rem}

\subsubsection{Asymptotically constant field and asymptotically flat surface: essentially isolated Landau levels}
Now assume that the magnetic field and scalar curvature are only asymptotically constant and asymptotically zero on $M$. By this, we mean that $R,\Theta\in C_\flat^\infty(M)$ satisfy
\begin{equation*}
\Theta=b+\Theta_{\rm pert},\quad R=R_{\rm pert},\qquad \Theta_{\rm pert}, R_{\rm pert}\in C_0^\infty(M),
\end{equation*}
where $C_0^\infty(M)$ denotes the smooth functions on $M$ vanishing at infinity.

\begin{prop}\label{prop:ess.spec}
Let $M$ be a noncompact complete Riemannian spin surface, which has scalar curvature $R\in C_\flat^\infty(M)$ satisfying $R=R_{\rm pert}\in C_0^\infty(M)$. Suppose the magnetic field $\Theta\in C_\flat^\infty(M)$ has the form $\Theta=b+\Theta_{\rm pert}$, with $0\neq b\in\RR$ and $\Theta_{\rm pert}\in C_0^\infty(M)$. Then $\sigma_{\rm ess}(H_\Theta)\subset(2\NN+1)|b|$, where $\sigma_{\rm ess}$ denotes the essential spectrum. If the Dirac coarse index ${\rm Ind}(D)$ is non-zero, then $\sigma_{\rm ess}(H_\Theta)=(2\NN+1)|b|$.
\end{prop}
\begin{proof}
Without loss of generality, assume $b>0$. Take $N$ to be any point in $M$, and apply the first part of Theorem \ref{thm:ASS.relative}. This says that the first part of Prop.\ \ref{prop:LL.iteration} holds, i.e.,
\begin{equation*}
\sigma(\varpi(\mathcal{H}_\Theta))\subset (2\NN+1)b,
\end{equation*}  
with $\varpi:C^*(M)\to C^*(M)/C^*_M(N)$ the quotient map.
As mentioned in Remark \ref{rem:localized.is.compact}, $C^*_M(N)$ is isomorphic to the compact operators on $L^2(M)$. Thus $\sigma(\varpi(\cdot))$ is the spectrum modulo compacts, i.e., the essential spectrum.

Now suppose ${\rm Ind}(D)\neq 0$ in $K_0(C^*(M))$. Since $C^*_M(N)$ comprises the compact operators, elements of $K_0(C^*_M(N))$ are represented by finite-rank projections. Thus the inclusion $K_0(C^*_M(N))\to K_0(C^*(M))$ is the zero map by Lemma \ref{lem:finite.projection.trivial}. Therefore, $\varpi_*:K_0(C^*(M))\to K_0(C^*(M)/C^*_M(N))$ is injective, so $\varpi_*\mathrm{Ind}(D)$ remains nonzero. Hence the second part of Theorem \ref{thm:ASS.relative} applies, and the second conclusion of Prop.\ \ref{prop:LL.iteration} holds: $(2\NN+1)b=\sigma(\varpi(\mathcal{H}_\Theta))=\sigma_{\rm ess}(H_\Theta)$. 
\end{proof}

\begin{rem}
Prop.\ \ref{prop:ess.spec} recovers, as a special case, the stability result with respect to $\Theta_{\rm pert}$ for $M$ the Euclidean plane \cite{Iwatsuka}. For the hyperbolic plane, a corresponding stability result \cite{Inahama} for the (finitely many) hyperbolic Landau levels can be deduced with a modified bootstrap method as detailed in \cite{LT}. The stability of the Landau levels against \emph{metric} perturbations $R_{\rm pert}\not\equiv 0$ appears to be a new result.
\end{rem}

\section{Landau operators on helical surfaces}
An instructive example of a Riemannian surface $M$ with scalar curvature \mbox{$R\not\in C_0^\infty(M)$} is the helicoid $X_c$, which will occupy us for the rest of this paper.

\subsection{Embedded helicoids}\label{sec:helicoid.definition}
The \emph{helicoid} $X_c$ with non-zero twisting parameter $c\in\RR$ is defined as the 2D submanifold of Euclidean $\RR^3$,
\begin{equation*}
X_c=\{(r\cos c\phi,r\sin c\phi,\phi)\in\RR^3\,:\,(r,\phi)\in\RR^2\},
\end{equation*}
oriented by $dr\wedge d\phi$, see Fig.\ \ref{fig:helicoids}. When $c>0$ (resp.\ $c<0$), the helicoid is right-handed (resp.\ left-handed). The metric on $X_c$ may be computed to be $dr^2+(1+c^2r^2)d\phi^2$; for instance, see pp.\ 96--97 and pp.\ 206 of \cite{DoCarmo}.

While $X_c$ is diffeomorphic to $\RR^2$, it is not isometric to a Euclidean plane, but only conformal to it. In \emph{isothermal coordinates} $(\rho,\phi)$ defined by
\begin{equation*}
\qquad \rho = \frac{\sinh^{-1}(cr)}{c}, \qquad \text{or}\qquad r = \frac{\sinh(c \rho)}{c},
\end{equation*}
the metric tensor is 
\begin{equation*}
g_{\rho\rho}=g_{\phi\phi}=\cosh^2(c\rho),\qquad g_{\rho\phi}=0=g_{\phi\rho}.
\end{equation*}
The volume form and scalar curvature are
\begin{equation*}
\omega_c=\cosh^2 (c\rho)\,d\rho\wedge d\phi,\qquad R_c(\rho,\phi)=R_c(\rho)=-\frac{2c^2}{\cosh^4(c\rho)}.
\end{equation*}
Note that as $c\rightarrow 0$, the coordinate transformation reduces to $\rho=r$, and the above formulae become those of the Euclidean $x$-$z$ plane. We also have $R_c\in C_\flat^\infty(X_c)$.

\subsubsection{Helicoid Landau operator for constant external field}\label{sec:embedded}
On Euclidean $\RR^3$, let $B=b\,dx\wedge dy$ be the Faraday 2-form for an externally applied magnetic field along the $z$-direction with strength $b\in \RR$, $b\neq 0$. Its restriction to $X_c$ is easily computed to be
\begin{equation}
B|_{X_c}(r,\phi)=bcr\,dr\wedge d\phi=\underbrace{\frac{bcr}{\sqrt{1+c^2r^2}}}_{\Theta_B(r,\phi)}\omega_c=\underbrace{b\tanh(c\rho)}_{\Theta_B(\rho,\phi)}\omega_c.\label{eqn:intrinsic.extrinsic}
\end{equation}
Here, the scalar function $\Theta_B$ is the \emph{intrinsic} field strength on the helicoid, measured against the volume form $\omega_c$ on $X_c$. Note that $\Theta_B\in C_\flat^\infty(X_c)$.

As $\rho\rightarrow\pm\infty$, the field strength becomes intrinsically constant, $\Theta_B\rightarrow \pm b\cdot{\rm sgn}(c)$, whereas at $\rho=0$, we have $\Theta_B\rightarrow 0$. Despite $B$ being a constant field on $\RR^3$, $\Theta_B$ is far from constant on $X_c$, and is more like a ``domain-wall'' magnetic field strength when viewed intrinsically.

\subsection{Half-helicoid and screw dislocation}\label{sec:half-helicoid}

Let $\mathring{\RR}^2$ be the plane with a small disk of small radius $r_0>0$ removed. The \emph{half-helicoid},
\begin{align*}
X_c^+&:=\{(r\cos(c\phi), r\sin(c\phi), \phi)\in\RR^3\,:\,r\geq r_0, \phi\in\RR\}\\
&=\{(r,\phi)\in X_c\,:\,r \geq r_0\},
\end{align*}
is a universal cover of $\mathring{\RR}^2$, embedded as a submanifold of $\RR^3$. \emph{It models a surface with a screw dislocation along the $z$-axis}, see Fig.\ \ref{fig:helicoids}. The boundary $\partial X_c^+$ is a helix parametrized by $\phi$.

\subsubsection{Landau operator on screw dislocated surface}
In the dislocation-free setting, we would have a family of parallel horizontal planes in $\RR^3$, on which the constant external field $B$ is also intrinsically constant (meaning that its restriction to the planes is a constant multiple of the volume form). 
These Euclidean planes' Landau operator $H_b^{(\RR^2)}$ will each have spectrum $(2\NN+1)|b|$, by Prop.\ \ref{prop:Euclidean.LL}. In physics language, such a ``stacking'' of 2D quantum Hall systems gives rise to a \emph{weak} topological insulator in 3D. 

On the screw dislocated surface $X_c^+$, the restriction of $B$ is no longer intrinsically constant, but is instead given by Eq.\ \eqref{eqn:intrinsic.extrinsic},
\begin{equation*}
\Theta_B(\rho,\phi)=b\tanh(c\rho),\;\; \rho \geq \rho_0.
\end{equation*}
Consequently, we should study the Landau operator $H_{\Theta_B}^{(X_c^+)}$, with the \emph{non-constant} intrinsic field strength $\Theta_B$, rather than $H_b^{(X_c^+)}$.

\begin{dfn}\label{dfn:dislocated.Landau}
The screw-dislocated Landau operator is the operator $H_{\Theta_B}^{(X_c^+)}$ on the half-helicoid $X_c^+$, made self-adjoint with either Dirichlet or Neumann boundary conditions.
\end{dfn}

\subsubsection{Auxiliary ``bulk'' Landau operator on $X_c$.}
Note that $\Theta_B$ is an \emph{odd} extension of $\Theta_B|_{X_c^+}$ to a function on $X_c$. It is actually more useful to consider, as a fictitious field strength on $X_c$, a smooth \emph{even} extension $\Theta_{B,{\rm ev}}\in C_\flat^\infty(X_c)$ of $\Theta_B|_{X_c^+}$, meaning that 
\begin{equation} 
\Theta_{B,{\rm ev}}(\rho,\phi)=b\tanh(c|\rho|),\qquad |\rho|>\rho_0,\label{eqn:field.fictitious}
\end{equation}
because then the limits as $\rho\rightarrow \pm\infty$ are both equal to {$b\cdot\mathrm{sgn}(c)$}. The extension in the $|\rho|<\rho_0$ region can be chosen arbitrarily. We shall think of $H_{\Theta_{B,{\rm ev}}}^{(X_c)}$ as the fictitious ``bulk Landau operator'' on the full helicoid $X_c$. Via a kind of \emph{bulk-boundary correspondence}, we will be able to relate the spectrum of $H_{\Theta_{B,{\rm ev}}}^{(X_c)}$ with that of $H_{\Theta_B}^{(X_c^+)}$.

\begin{rem}
Our model, with $X_c^+$ embedded in $\RR^3$ and constant external magnetic vector field in the $z$-direction, formalizes what physicists have in mind when they think of helically propagating states induced by a screw dislocation. In fact, a time-reversal symmetric version of such helical modes has recently been experimentally realized in acoustic topological insulators embedded in three dimensions \cite{Xue,Ye}.

A simpler related model is the Landau operator $H_b^{(X_c)}$ on the helicoid $X_c$ with an intrinsically constant magnetic field $b\cdot\omega_c=b\cdot \cosh^2(c\rho)d\rho\wedge d\phi$. The spectrum for such a model is easier to analyze using the identity Eq.\ \eqref{eqn:bootstrap1}. Indeed, for $c>0$, the operator $H_b^{(X_c)}$ coincides at large $|\rho|$ with our fictitious ``bulk Landau operator'' $H_{\Theta_{B,{\rm ev}}}^{(X_c)}$  However, if we were to realize $H_b^{(X_c)}$ on the embedded surface $X_c\subset\RR^3$, the perpendicular magnetic vector field would rotate wildly at points near to the axis (i.e., at small $|\rho|$), which is rather unrealistic. For large $|\rho|$, the approximately vertical magnetic vector field would also switch directions when $\rho$ is replaced by $-\rho$.
\end{rem}

\subsection{Delocalized index for helicoid Landau operators}\label{sec:delocalized.index}
Because $\Theta_{B,{\rm ev}}$ and $R_c$ are \emph{not} asymptotically constant, Prop.\ \ref{prop:ess.spec} does not apply, and we do not have isolated Landau levels even when considered inside the essential spectrum of $H_{\Theta_B}^{(X_c)}$. Nevertheless, we at least have $R_c$ vanishing away from the helix $\partial X_c^+$, and the field strength $\Theta_{B,{\rm ev}}$ can be considered a ``horizontal perturbation'' of the constant function $\mathrm{sgn}(c)\cdot b$ in the following sense.

\begin{dfn}\label{dfn:h.pert}
A \emph{horizontally perturbed} field strength $\Theta\in C_\flat^\infty(X_c)$ is one which differs from a constant by a $\Theta_{\rm pert}$ with $\lim_{\rho\rightarrow\pm\infty}\Theta_{\rm pert}(\rho,\phi)=0$.
\end{dfn}

We will apply Theorem \ref{thm:ASS.relative} and Prop.\ \ref{prop:LL.iteration} to $M=X_c$ and $N=\partial X_c^+$. Thus, we analyze the localization exact sequence
\begin{equation}
0\longrightarrow C^*_{X_c}(\partial X_c^+)\longrightarrow C^*(X_c)\overset{\varpi}{\longrightarrow}C^*(X_c)/C^*_{X_c}(\partial X_c^+)\longrightarrow 0.\label{eqn:quotient.SES}
\end{equation}
The $K$-theory long exact sequence for \eqref{eqn:quotient.SES} is
\begin{align*}
\cdots \rightarrow \underbrace{K_0(C^*_{X_c}(\partial X_c^+))}_{0}\rightarrow \underbrace{K_0(C^*(X_c))}_{\ZZ} &\overset{\varpi_*}{\rightarrow} K_0(C^*(X_c)/C^*_{X_c}(\partial X_c^+))  \\
&\quad \rightarrow \underbrace{K_1(C^*_{X_c}(\partial X_c^+))}_{\ZZ}\rightarrow \underbrace{K_1	(C^*(X_c))}_{0}\rightarrow \cdots
\end{align*}
so that $K_0(C^*(X_c)/C^*_{X_c}(\partial X_c^+))\cong \ZZ\oplus \ZZ$. Here we used a general result that $K_\bullet(C^*_{X_c}(\partial X_c^+))\cong K_\bullet(C^*(\partial X_c^+))$, and the observation that the helix $\partial X_c^+$ is isometric to a Euclidean line. We also used the fact that $K_0(C^*(X_c))\cong K_0(X_c)\cong\ZZ$, which follows from the coarse Baum--Connes conjecture being verified for $X_c$, since it is non-positively curved and simply-connected (Corollary 7.4 of \cite{HR-coarse}). Moreover, since the Dirac operator $D$ on $X_c$ represents the generator of $K_0(X_c)$, its coarse index ${\rm Ind}(D)$ is also a generator of $K_0(C^*(X_c)) \cong \ZZ$. By the injectivity of $\varpi_*$ obtained by the above exact sequence, we have verified that:
\begin{lem}\label{lem:coarse.index.helicoid}
The quotient coarse index $\varpi_*({\rm Ind}(D))$ of the helicoid Dirac operator is non-zero in $K_0(C^*(X_c)/C^*_{X_c}(\partial X_c^+))$. 
\end{lem}

We now determine the spectrum of the helicoid Landau operator considered in the quotient Roe algebra $C^*(X_c)/C^*_{X_c}(\partial X_c^+)$, generalizing Prop.\ \ref{prop:ess.spec}.

\begin{thm}\label{thm:quotient.Landau}
For a horizontally perturbed field $\Theta=b+\Theta_{\rm pert}, b\neq 0$, the spectrum of $\varpi(\mathcal{H}_\Theta^{(X_c)})$ is $(2\NN + 1)|b|$.
Moreover, for each $n\in \NN$, the class 
\begin{equation}
  \left[P_{(2n+1)|b|}\bigl(\varpi(\mathcal{H}_\Theta^{(X_c)})\bigr)\right] ~~\in~~ K_0(C^*(X_c)/C^*_{X_c}(\partial X_c^+))\label{eqn:delocalized.index}
\end{equation}
coincides, up to a possible sign, with the nonzero class  $\varpi_*{\rm Ind}(D)$, where ${\rm Ind}(D)\in K_0(C^*(X_c))$ is the coarse index of the Dirac operator on $X_c$.
\end{thm}
\begin{proof}
The helicoid curvature $R_c$ and a horizontally perturbed field $\Theta$ satisfy the assumptions of Theorem \ref{thm:ASS.relative}, for $N = \partial X_c^+$. This, together with Lemma \ref{lem:coarse.index.helicoid}, means that Prop.\ \ref{prop:LL.iteration} applies, and the result follows immediately.
\end{proof}

\begin{rem} 
Theorem \ref{thm:quotient.Landau} applies, in particular, to $H_{\Theta_{B,{\rm ev}}}^{(X_c)}$.
It makes precise the idea that the Landau levels $(2n+1)|b|$ still have a well-defined ``delocalized index'' $[P_{(2n+1)|b|}(\varpi(\mathcal{H}_\Theta^{(X_c)}))]$,
despite these spectral values generally not being isolated in the full (or even essential) spectrum. In passing to the quotient algebra via $\varpi$, we have ``discarded the spectral data localized near $\partial X_c^+$''.
\end{rem}

The projection $P_{(2n+1)|b|}(\varpi(\mathcal{H}_\Theta^{(X_c)}))$ of Theorem \ref{thm:quotient.Landau} can be obtained by continuous functional calculus as follows. For each $n\in\NN$, let $\varphi_n\in C_0(\RR)$ be a bump function for the $n$-th Landau level, in the sense that
\begin{align}
\varphi_n((2n+1)|b|)&=1,\nonumber\\
{\rm supp}(\varphi_n)&\subset ((2n-1)|b|,(2n+3)|b|). \label{eqn:LL.bump.function}
\end{align}
Observe that $\varphi_n$ may fail to be idempotent only within the gaps between the \mbox{$n$-th} Landau level and the adjacent ones. Since $\varpi(\mathcal{H}_\Theta^{(X_c)})$ only has spectrum at the Landau levels, we may write
\begin{equation*}
P_{(2n+1)|b|}\bigl(\varpi(\mathcal{H}_\Theta^{(X_c)})\bigr)=\varphi_n\bigl(\varpi(\mathcal{H}_\Theta^{(X_c)})\bigr)
\end{equation*}
for any choice of bump function as above.
Note that without passing to the quotient, the operator $\varphi_n(H_\Theta^{X_c})\in C^*(X_c)$ is not yet a projection, due to the possibly non-trivial support of $\varphi_n^2-\varphi_n$ in $\sigma(H_\Theta^{(X_c)})$. Rather, we have
\begin{equation}
\varpi\bigl(\varphi_n(H_\Theta^{(X_c)})\bigr)=\varphi_n\bigl(\varpi(\mathcal{H}_\Theta^{(X_c)})\bigr)=P_{(2n+1)|b|}\bigl(\varpi(\mathcal{H}_\Theta^{(X_c)})\bigr)\in C^*(X_c)/C^*_{X_c}(\partial X_c^+).\label{eqn:quotient.projection.calculus}
\end{equation}

\section{Gaplessness of screw-dislocated Landau operator}
The bulk helicoid Landau operator $H_{\Theta_{B,{\rm ev}}}^{(X_c)}$ will be related to the screw-dislocated Landau operator $H_{\Theta_B}^{(X_c^+)}$ (Definition \ref{dfn:dislocated.Landau}) through a coarse Mayer--Vietoris sequence.

For ease of notation, we will simply write $\Theta$ for $\Theta_{B,{\rm ev}}$ on the helicoid $X_c$, and also $\Theta=\Theta_B=\Theta_B|_{X_c^+}$ on the half-helicoid $X_c^+$.

\subsection{Coarse Mayer--Vietoris}
Consider the partition $X_c=X_c^+\cup X_c^-$ where $X_c^+$ is our half-helicoid surface and 
\begin{equation*}
X_c^-=\{(\rho,\phi)\in X_c\,:\,\rho\leq \rho_0\}.
\end{equation*}
The intersection $X_c^+\cap X_c^-=\partial X_c^+=-\partial X_c^-$ is the helix $\{\rho=\rho_0\}$. The coarse Mayer--Vietoris (MV) sequence \cite{HRY} for this (coarsely excisive) partition has boundary map 
\begin{equation}
\partial_{\rm MV}:K_0(C^*(X_c))\overset{\cong}{\rightarrow} K_1(C^*_{X_c}(\partial X_c^+)),\label{eqn:helicoid.MV}
\end{equation}
which can be computed to be an isomorphism $\ZZ\rightarrow \ZZ$, as follows. As in the proof of Lemma 3.3 of \cite{LT}, we just need the information that $K_0(C^*(X_c))\cong\ZZ$, $K_1(C^*(X_c))\cong 0$ (Baum--Connes), $K_0(C^*(\partial X_c^+))\cong 0$, $K_1(C^*(\partial X_c^+))\cong \ZZ$ (isometry of $\partial X_c^+$ with $\RR$), and the observation that the reflection $\rho\mapsto -\rho$ induces an isomorphism $K_\bullet(C^*(X_c^+))\cong K_\bullet(C^*(X_c^-)), \bullet=0,1$.

The restriction map $r:C^*(X_c)\rightarrow C^*(X_c^+)$ is a homomorphism up to terms in $C^*(X_c^+)$ localized near the boundary $\partial X_c^+$. Thus we have a restriction morphism $\tilde{r}$ to the quotient Roe algebra,
\begin{equation*}
\tilde{r}:C^*(X_c)\rightarrow C^*(X_c^+)/C^*_{X_c^+}(\partial X_c^+),
\end{equation*}
which actually factors through $\varpi$ (defined in Eq.\ \eqref{eqn:quotient.SES}),
\begin{equation}
\tilde{r}:C^*(X_c)\overset{\varpi}{\rightarrow}C^*(X_c)/C^*_{X_c}(\partial X_c^+)\overset{\hat{r}}{\rightarrow} C^*(X_c^+)/C^*_{X_c^+}(\partial X_c^+).\label{eqn:factorisation.quotient}
\end{equation}
There is a short exact sequence
\begin{equation*}
0\rightarrow C^*_{X_c^+}(\partial X_c^+) \rightarrow C^*(X_c^+)\overset{q}{\rightarrow} C^*(X_c^+)/C^*_{X_c^+}(\partial X_c^+)\rightarrow 0,\label{eqn:SES}
\end{equation*}
whose long exact sequence has a connecting map
\begin{equation}
\delta:K_0(C^*(X_c^+)/C^*_{X_c^+}(\partial X_c^+))\rightarrow K_1(C^*_{X_c^+}(\partial X_c^+))\cong K_1(C^*_{X_c}(\partial X_c^+)). \label{eqn:connecting.defn}
\end{equation}
This connecting map and the MV boundary map are related (Prop.\ 1.3 of \cite{LT}), by
\begin{equation}
\partial_{\rm MV}=\delta\circ\tilde{r}_*.\label{eqn:MV.equals.connecting}
\end{equation}

\medskip

Let $\varphi_n$ be a bump function for the $n$-th Landau level, as in Eq.\ \eqref{eqn:LL.bump.function}. The operators $\varphi_n(H_\Theta^{(X_c^+)})$ and $r(\varphi_n(H_{{\Theta}}^{(X_c)}))$ are both elements of $C^*(X_c^+)$, and their difference is contained in the localized Roe algebra $C^*_{X_c^+}(\partial X_c^+)$ (Lemma 1.7 of \cite{LT}). 
So by passing to quotients, we have an equality
\begin{align*}
q\left(\varphi_n(H_\Theta^{(X_c^+)})\right)&=\tilde{r}\Bigl(\underbrace{\varphi_n(H_{{\Theta}}^{(X_c)})}_{\text{non-projection}}\Bigr)\\
&\overset{\eqref{eqn:factorisation.quotient}}{=} \hat{r}\Bigl(\underbrace{\varpi(\varphi_n(H_{{\Theta}}^{(X_c)}))}_{\rm projection\;by\;Eq.\ \eqref{eqn:quotient.projection.calculus}}\Bigr)\;\;\in C^*(X_c^+)/C^*_{X_c^+}(\partial X_c^+) 
\end{align*}
from which we deduce that the left hand side is in fact a projection.
Now 
\begin{align}
\delta[q(\varphi_n(H_\Theta^{(X_c^+)})]&=\delta\left[\tilde{r}\left(\varphi_n(H_{{\Theta}}^{(X_c)})\right)\right]\nonumber\\
&=\delta\left(\hat{r}_*\left[\varpi\left(\varphi_n(H_\Theta^{(X_c)})\right)\right]\right)\nonumber\\
&=\delta\circ\hat{r}_*\left[P_{(2n+1)|b|}\left(\varpi(\mathcal{H}_\Theta^{(X_c)})\right)\right] & &({\rm Eq.}\, \eqref{eqn:quotient.projection.calculus}) \nonumber\\
&=\pm \delta\circ\hat{r}_*\circ\varpi_*({\rm Ind}(D)) & &({\rm Theorem}\; \ref{thm:quotient.Landau})\nonumber\\
&=\pm \delta\circ\tilde{r}_*({\rm Ind}(D))\nonumber\\
&=\pm \partial_{\rm MV}({\rm Ind}(D)) & &({\rm Eq.}\,\eqref{eqn:MV.equals.connecting})\nonumber\\
&\neq 0 \in K_1(C^*_{X_c^+}(\partial X_c^+)) & & ({\rm Eq.}\,\eqref{eqn:helicoid.MV}),\label{eqn:nontrivial.connecting.map}
\end{align}
where the non-vanishing in the last line is due to ${\rm Ind}(D)$ being a generator of $K_0(C^*(X_c))$.

\subsection{Gap-filling argument}\label{sec:gap-filling}
Now $\varphi_n(H_\Theta^{(X_c^+)})$ cannot be a projection, otherwise exactness of the long exact sequence (i.e., $\delta\circ q_*=0$) will lead to the contradiction
\begin{equation*}
0=\delta\circ q_*[\varphi_n(H_\Theta^{(X_c^+)})]
\overset{{\rm Eq.}\, \eqref{eqn:nontrivial.connecting.map}}{\neq} 0.
\end{equation*}
Thus $H_\Theta^{(X_c^+)}$ must have some spectrum in the support of $\varphi_n^2-\varphi_n$. 

Consider the $n=0$ case first. The support of $\varphi_0^2-\varphi_0$ can be chosen to lie below the semibounded spectrum of $H_\Theta^{(X_c^+)}$, and inside any subinterval of $(|b|,3|b|)$, so spectra must appear in the latter subinterval. By varying the choice of subinterval, we see that the \emph{entire} interval $(|b|,3|b|)$ must be filled with spectrum of $H_\Theta^{(X_c^+)}$, to successfully prevent $\varphi_0(H_\Theta^{(X_c^+)})$ from being a projection. This shows gap-filling for the interval between the $0$-th and $(0+1)$-th Landau levels.

For $n\geq 1$, consider $\tilde{\varphi}_n=\sum_{k=0}^n \varphi_k$, with the $\varphi_k$ chosen to have disjoint supports. Then $\varpi(\tilde{\varphi}_n(H_\Theta^{(X_c)}))$ is a direct sum of projections. By adding to $\tilde{\varphi}_n$ a suitable supplementary function $\psi_n$ supported within $(|b|,(2n+1)|b|)$, we can arrange for $\tilde{\varphi}_n+\psi_n$ to have value 1 on the interval $[|b|,(2n+1)|b|]$. This extra $\psi_n$ is supported away from the Landau levels, so $\varpi(\psi_n(H_\Theta^{(X_c)}))=\psi_n(\varpi(\mathcal{H}_\Theta^{(X_c)}))=0$ by Theorem \ref{thm:quotient.Landau}. Thus $\varpi((\tilde{\varphi}_n+\psi_n)(H_\Theta^{(X_c)}))=\varpi(\tilde{\varphi}_n(H_\Theta^{(X_c)}))$ is left intact. 
Its $K$-theory class is again non-trivial (by additivity under direct sums), and we derive the same contradiction forbidding $\tilde{\varphi}_n(H_\Theta^{(X_c^+)})$ from being a projection. Now we deduce that the gap between the $n$-th and $(n+1)$-th Landau levels must be completely filled with spectrum of $H_\Theta^{(X_c^+)}$. By induction, we deduce our main result (restoring the subscript in $\Theta_B$):
\begin{thm}\label{thm:helical.gap.filling}
The screw-dislocated Landau operator $H_{\Theta_B}^{(X_c^+)}$ (Definition, \ref{dfn:dislocated.Landau}) has no gaps in its spectrum above the LLL $|b|$.
\end{thm}

\medskip

\begin{rem}\label{rem:current}
The $K_1$-class $\delta[q(\varphi_n(H_\Theta^{(X_c^+)}))]$, which obstructs the existence of spectral gaps for $H_\Theta^{(X_c^+)}$, has a more refined interpretation using the methods of \cite{LTcobordism} Section 6. Technically, this interpretation requires a polynomial growth condition, which is satisfied by helicoids. In brief, let $X_c^{+,\uparrow}=\{(\rho,\phi)\in X_c^+\,:\,\phi\geq 0\}$ be the upper half of $X_c^+$, and $X_c^{+,\downarrow}$ be the lower half. Then there is a well-defined integer-valued map, Definition 4.6 of \cite{LTcobordism},
\begin{equation*}
\theta_{X_c^{+,\uparrow}}:K_1(C^*_{X_c^+}(\partial X_c^+))\rightarrow \ZZ,\qquad [u]\mapsto {\rm Index}\, T_u,
\end{equation*}
where $T_u$ is the compression of the unitary $u\in \mathrm{M}_n(C^*_{X_c^+}(\partial X_c^+)^+)$ to $X_c^{+,\uparrow}$. When $\theta_{X_c^{+,\uparrow}}$ is applied to the obstruction class
$\delta[q(\tilde{\varphi}_n(H_\Theta^{(X_c^+)}))]$, the resulting integer has an interpretation as a quantized current channel flowing from $X_c^{+,\uparrow}$ to $X_c^{+,\downarrow}$, provided by the generalized eigenstates of $H_\Theta^{X_c^+}$ with energies lying within the gap between the $n$-th and $(n+1)$-th Landau levels. This channel receives no contribution from the ``very delocalized'' Landau levels, and may therefore be thought of as arising from gap-filling ``helical edge states'' localized near the boundary helix $\partial X_c^+$.
\end{rem}

\medskip

\begin{rem}\label{rem:robust}
Much like the hyperbolic plane case analyzed in \cite{LT}, the coarse geometry methods mean that the gaplessness result in Theorem \ref{thm:helical.gap.filling} is very robust:
\begin{itemize}
\item The projected radius $\rho_0$ of the boundary helix is arbitrary. More generally, the geometry of the boundary $\partial X_c^+$ can be significantly adjusted, as long as the coarse MV calculations remain intact. In particular, no $\phi$-translation symmetry is required of $\partial X_c^+$.
\item The Dirichlet/Neumann boundary condition on $\partial X_c^+$ can be generalized considerably, see Remark 1.8 of \cite{LT}.
\item The externally applied field $B$ can be horizontally perturbed from a constant.
\item A bounded confining potential which decays away from $\partial X_c^+$ can be added, and treated as a perturbation in the same way as $\Theta_{\rm pert}$ was.
\item The embedding of the helical surface $X_c^+$ in $\RR^3$ can also be modified, as long as the resulting curvature and intrinsic field strength induced on $X_c^+$ can be treated as a horizontal perturbation. In particular, the ``twisting rate'' parameter $c\neq 0$ can be changed without destroying the gaplessness. 
\end{itemize}
\end{rem}

\medskip

\subsection{Comparison with a discrete model approach}\label{sec:relation.discrete}
The construction of a delocalized coarse index in \S \ref{sec:delocalized.index} and the gap-filling argument of this section are parallel to the proof of the \emph{bulk--dislocation correspondence} for \emph{discrete} models of $3$-dimensional topological insulators given by the first author \cite{Kubota}. 
That work dealt with a $3$-dimensional discrete Hamiltonian operator $H$ representing a so-called weak topological phase, acting on a lattice with screw dislocation (in other words, a certain discretization of the helical surface). 
More precisely, consider a periodic $3$-dimensional type A topological insulator having non-trivial weak topology in the $xy$-direction.
If a screw dislocation is inserted along the $z$-axis into the configuration of atoms of the same material, then the spectral gap is filled by localized states near the $z$-axis. 

This is proved in a similar way as Theorem \ref{thm:helical.gap.filling}. 
Indeed, in the ``layered'' case where $H$ is the direct sum of countably infinite copies of a $2$-dimensional topological insulator (thought of being stacked on top of one another),  the corresponding screw-dislocated Hamiltonian, denoted by $\widetilde{H}$, is viewed as an operator on the (discretized) half helicoid $X_c^+$.
This $\widetilde{H}$ determines a self-adjoint element of $C^*(X_c^+)$, and its image in $C^*(X_c^+)/C^*_{X_c^+}(\partial X_c^+)$ has a spectral gap, as does $H$. 
A similar coarse Mayer--Vietoris argument as Eq.\ \eqref{eqn:nontrivial.connecting.map} shows that the connecting map \eqref{eqn:connecting.defn} is non-trivial on a weak topological insulator invariant. 

Although the ideas of our gap-filling argument and \cite{Kubota} are similar, there are some essential differences.
First, \cite{Kubota} takes $3$-dimensional (discrete model) Hamiltonians as input. Indeed, a $3$-dimensional gapped Hamiltonian which is not layered (due to interaction terms between different layers) may be homotopic to a layered one. The gap-filling argument is applicable to such Hamiltonians as well. 
On the other hand, in the continuum setting, the direct sum of infinitely many layers of $2$-dimensional Landau operators is not a continuum Hamiltonian on a $3$-dimensional manifold.

Second, the lifting argument of operators, which is a central ingredient in \cite{Kubota}, does not appear in this paper. 
Rather than starting  from the operator $H_{\Theta}^{(X_c)}$ on the helicoid, \cite{Kubota} first considers a Hamiltonian on (discretized) $\RR^2$ and then lifts it onto the (discretized) helicoid $X_c^+$ modulo the $z$-axis. 
This lifting argument is realized at the Roe algebra level, by the $\ast$-homomorphism
\[s \colon C^*(|\mathbb{R}^2|) \to C^*(X_c^+)^\ZZ /C^*_{ X_c^+}(\partial X_c^+)^\mathbb{Z} ,\]
which we called the codimension 2 transfer map.

\section{Fourier transform approach}\label{sec:Fourier}

\emph{Under a $\phi$-invariance assumption}, which is rather restrictive from the physical perspective, Theorem \ref{thm:helical.gap.filling} can also be understood from a more functional analytic viewpoint, in terms of the spectral flow of catenoid Landau operators.

In the first instance, we now need to restrict to horizontally perturbed field strengths $\Theta=b+\Theta_{\rm pert}$ which are $\phi$-invariant, and work in the \emph{Landau gauge}. The latter means that the connection 1-form $A$ with curvature $\Theta\cdot\omega_c$ has the form
\begin{equation*}
A(\rho,\phi)=a_\phi(\rho)\,d\phi,\qquad a_\phi\in C^\infty(\RR).
\end{equation*}
Explicitly, we take 
\begin{equation}
a_\phi(\rho)=\int_0^\rho \Theta(\rho^\prime)\cosh^2(c\rho^\prime)\,d\rho^\prime.\label{eqn:Landau.gauge.potential}
\end{equation}
In this gauge, the helicoid Landau operator $H_{\Theta}^{(X_c)}=(d-iA)^*(d-iA)$ is translation invariant in $\phi$.

\subsection{Reduction to loop of catenoid Landau operators}\label{sec:catenoid}
Let $\tau_n, n\in\ZZ$ denote the action of $\ZZ$ on $X_c$ by the deck transformations $\tau_n:\phi\mapsto \phi+2n\pi/|c|$. The quotient of $X_c$ by this action is the \emph{catenoid} $Y_c$. Similarly the quotient of $X_c^+$ is the half-catenoid $Y_c^+$, and the quotient of the boundary helix $\partial X_c^+$ is a circle $\partial Y_c^+$.
\medskip

The Fourier transform of $H_\Theta^{(X_c)}$, with respect to the $\ZZ$-action $\tau$, is a family of \emph{catenoid} Landau operators labelled by the character $e^{ik}: n\mapsto e^{ink}$,
\begin{equation}
H_\Theta^{(X_c)}\cong \int_{e^{ik}\in{\rm U}(1)=\widehat{\ZZ}}^\oplus H_\Theta^{(Y_c)}(e^{ik}). \label{eqn:helicoid.Fourier}
\end{equation}
As in Bloch--Floquet theory, $H_\Theta^{(Y_c)}(e^{ik})$ denotes the operator $H_\Theta^{(X_c)}$ acting on functions of $X_c$ subject to quasiperiodic boundary conditions $f(\phi+\frac{2\pi}{|c|})=e^{ik}f(\phi)$, with inner product given by that on a fundamental domain  
\begin{equation} \label{EqDefUc}
U_c:=\{ (\rho,\phi) \in X_c \mid 0 \leq \phi < 2\pi /|c|\}.
\end{equation}
Such an $f$ can be viewed as a function on $Y_c$, but twisted by a flat line bundle with holonomy $e^{ik}$. To identify the Hilbert spaces $L^2(Y_c;e^{ik})$ at each $k\in\RR$, we use the (singular) gauge transforms
\begin{align*}
u_k : L^2(Y_c;e^{ik}) & \rightarrow L^2(Y_c) \\
u_k(f)(\rho,\phi)&=e^{-ik|c| \phi /2\pi}f(\rho,\phi).
\end{align*}
We have
\begin{equation}
    u_k H_\Theta^{(Y_c)}(e^{ik}) u_k^* = \Big(d-i\Big( A -  \frac{k|c|}{2\pi}d\phi \Big)\Big)^*\Big(d-i\Big( A -  \frac{k|c|}{2\pi}d\phi \Big)\Big). \label{eqn:catenoid_Landau}
\end{equation}
This identification will be used implicitly later on.

\begin{prop}\label{prop:continuity.field}
Let $\Theta(\rho,\phi)=\Theta(\rho)$ be a $\phi$-invariant magnetic field strength on the helicoid $X_c$. Under the identifications Eq.\ \eqref{eqn:catenoid_Landau}, the catenoid Landau Hamiltonians $H_\Theta^{(Y_c)}(e^{ik})$ depend norm-resolvent continuously on $k\in \RR$. 
Moreover, the resolvent difference 
\[
\left(H_\Theta^{(Y_c)}(e^{ik^\prime})+1\right)^{-1}-\left(H_\Theta^{(Y_c)}(e^{ik})+1\right)^{-1}
\]
 is compact, for all $k,k^\prime\in \RR$.
\end{prop}

\begin{proof}
From the discussion preceding the proposition, the connection 1-form for the catenoid Landau operator $H_\Theta^{(Y_c)}(e^{ik})$ in Landau gauge is
\begin{equation*}
A_k\equiv A_k(\rho)=\left(a_\phi(\rho)+\frac{k|c|}{2\pi}\right)d\phi,
\end{equation*} 
where $a_\phi=a_\phi(\rho)=\int_0^\rho \Theta(\rho^\prime)\cosh^2(c\rho^\prime) \,d\rho^\prime$ as in Eq.\ \eqref{eqn:Landau.gauge.potential}. We write $\nu^{-1}:=\cosh(c\rho)$ for the conformal factor, and observe that $\nu\equiv\nu(\rho)$ vanishes as $\rho\to\infty$.

In conformal coordinates $(\rho,\phi)$, we have $H_\Theta^{(Y_c)}(e^{ik})=-\nu^2(\partial_{\phi,A_k}^2+\partial_{\rho,A_k}^2)$. On $C_c^\infty(Y_c)$, the conformally scaled partial derivative, $-i\nu\partial_\phi$, is symmetric (with respect to the inner product on $L^2(Y_c)$), and so is its covariant version $-i\nu\partial_{\phi,A_k}$ in Landau gauge. We also have $\partial_\rho = \partial_{\rho, A_k}$ in Landau gauge, but $-i\nu\partial_\rho=-i\nu\partial_{\rho,A_k}$ is not symmetric because of the $\rho$-dependent $A_k$. Nevertheless, the quadratic term $-\nu^2\partial_{\rho}^2$ is positive definite, since
\begin{align*}
\langle f | -\nu^2 \partial_\rho^2 f\rangle_{L^2(Y_c)} &=-\int \overline{f(\rho,\phi)}\, \mathrm{sech}^2(c\rho)\,\partial_\rho^2f(\rho,\phi)\,\cosh^2(c\rho)\,d\rho\,d\phi\\
&=-\int \overline{f(\rho,\phi)}\,\partial_\rho^2f(\rho,\phi)\,d\rho\,d\phi\\
&=\int \overline{\partial_\rho f(\rho,\phi)}\partial_\rho f(\rho,\phi)\,d\rho\,d\phi\\
&=\int | \partial_\rho f(\rho,\phi)|^2\,d\rho\,d\phi\; \geq \; 0,\qquad\qquad f\in C_c^\infty(Y_c). 
\end{align*}
So, for $f\in C_c^\infty(Y_c)$, we have
\begin{align*}
\|-i\nu\partial_{\phi,A_k}f\|^2_{L^2(Y_c)} &= \langle f | \nu^2\partial_{\phi,A_k}^2 f\rangle_{L^2(Y_c)}\\
&\leq \langle f| H_\Theta^{(Y_c)}(e^{ik}) f\rangle_{L^2(Y_c)} \\
&\leq \langle f | (H_\Theta^{(Y_c)}(e^{ik})+1)f\rangle_{L^2(Y_c)} \\
&=\|(H_\Theta^{(Y_c)}(e^{ik})+1)^{1/2}f\|^2_{L^2(Y_c)}.
\end{align*}
Therefore, the bounded operator $\nu\partial_{\phi,A_k}\left(H_\Theta^{(Y_c)}(e^{ik})+1\right)^{-1/2}$, hence also the bounded operators
\begin{equation}
\nu\partial_{\phi,A_k}\left(H_\Theta^{(Y_c)}(e^{ik})+1\right)^{-1},\qquad \left(H_\Theta^{(Y_c)}(e^{ik})+1\right)^{-1}\nu\partial_{\phi,A_k}\label{eqn:bounded.argument}
\end{equation}
can be constructed uniquely, with norm at most 1. This works for every $k\in\RR$.

Now consider a pair of catenoid Landau operators labelled by $k,k^\prime\in\RR$. On their common core $C^\infty_c(Y_c)$, their difference is
\begin{align*}
&H_\Theta^{(Y_c)}(e^{ik^\prime})-H_\Theta^{(Y_c)}(e^{ik})\\
 &\qquad=-\nu^2\left(\partial_{\phi,A_{k^\prime}}^2-\partial_{\phi,A_k}^2\right)\\
&\qquad=-\nu^2\Bigg(\partial_{\phi,A_{k^\prime}}\partial_{\phi,A_k}-\partial_{\phi,A_{k^\prime}}\frac{i(k^\prime-k)|c|}{2\pi}-\partial_{\phi,A_{k^\prime}}\partial_{\phi,A_k}+\frac{i(k-k^\prime)|c|}{2\pi}\partial_{\phi,A_k}\Bigg)\\
&\qquad=i\left(\nu\partial_{\phi,A_{k^\prime}}+\nu\partial_{\phi,A_k}\right)\frac{\nu(k^\prime-k)|c|}{2\pi}.
\end{align*}
Their resolvent difference is then
\begin{align}
&\left(H_\Theta^{(Y_c)}(e^{ik})+1\right)^{-1}-\left(H_\Theta^{(Y_c)}(e^{ik^\prime})+1\right)^{-1}\nonumber
\\
&\qquad=\left(H_\Theta^{(Y_c)}(e^{ik^\prime})+1\right)^{-1}\left(H_\Theta^{(Y_c)}(e^{ik^\prime})-H_\Theta^{(Y_c)}(e^{ik})\right)\left(H_\Theta^{(Y_c)}(e^{ik})+1\right)^{-1}\nonumber\\
&\qquad=i\left(H_\Theta^{(Y_c)}(e^{ik^\prime})+1\right)^{-1}\nu\partial_{\phi,A_{k^\prime}}\cdot\frac{\nu(k^\prime-k)|c|}{2\pi}\left(H_\Theta^{(Y_c)}(e^{ik})+1\right)^{-1}\nonumber\\
&\qquad\;\;\;\;+i\left(H_\Theta^{(Y_c)}(e^{ik^\prime})+1\right)^{-1}\frac{\nu(k^\prime-k)|c|}{2\pi}\cdot\nu\partial_{\phi,A_k}\left(H_\Theta^{(Y_c)}(e^{ik})+1\right)^{-1}.\label{eqn:resolvent.difference}
\end{align}
Because the operators in Eq.\ \eqref{eqn:bounded.argument}
are bounded by 1 (independently of $k, k^\prime$), the formula Eq.\ \eqref{eqn:resolvent.difference} holds not only on $C^\infty_c(Y_c)$, but on all of $L^2(Y_c)$. We deduce that as $|k^\prime-k|\to 0$, the resolvent difference Eq.\ \eqref{eqn:resolvent.difference} tends to 0 in norm.

Finally, since $\nu$ vanishes at infinity while $(H_\Theta^{(Y_c)}(e^{ik})+1)^{-1}$ and $(H_\Theta^{(Y_c)}(e^{ik^\prime})+1)^{-1}$ are locally compact, we deduce that the resolvent difference, Eq.\ \eqref{eqn:resolvent.difference}, is compact. 
\end{proof}

\begin{cor}\label{cor:continuous.loop}
Let $\varphi\in C_0(\RR)$, and $\Theta$ be a $\phi$-invariant magnetic field strength. Under the Fourier transform \eqref{eqn:helicoid.Fourier}, the operator $\varphi(H_\Theta^{(X_c)})$ becomes a continuous loop 
\begin{equation*}
\widehat{\ZZ}=\mathrm{U}(1)\ni e^{ik}\mapsto \varphi\left(H_\Theta^{(Y_c)}(e^{ik})\right)\in C^*(Y_c).
\end{equation*}
\end{cor}
\begin{proof}
Due to Prop.\ \ref{prop:continuity.field}, for $k\in[0,2\pi]$, the path $k\mapsto u_k H_\Theta^{(Y_c)}(e^{ik}) u_k^*$ is norm-resolvent continuous, and so $k\mapsto u_k \varphi(H_\Theta^{(Y_c)}(e^{ik})) u_k^* \in C^*(Y_c)$ is norm-continuous. Conjugate by the continuous path of unitaries $k\mapsto u_k^*$ (in the multiplier algebra of $C^*(Y_c)$), so the endpoint operators are identified, then we obtain a continuous loop.
\end{proof}

Let
\begin{equation*}
\varpi':\mathcal{B}(L^2(Y_c)) \to  \mathcal{Q}(L^2(Y_c)):= \mathcal{B}(L^2(Y_c))/ \mathcal{K}(L^2(Y_c))
\end{equation*}
be the quotient map to the Calkin algebra.

\begin{lem}\label{lem:quotient.projection}
For any horizontally perturbed $\phi$-invariant bounded field strength $\Theta$, and any bump function $\varphi_n$ for the $n$-th Landau level (Eq.\ \eqref{eqn:LL.bump.function}), the family
\begin{equation*}
\left\{ \varpi' \big( \varphi_n\big( H_\Theta^{(Y_c)}(e^{ik}) \big)\big) \right\}_{e^{ik} \in\mathrm{U}(1)= \widehat{\ZZ}}
\end{equation*} defines a projection in $C(\widehat{\ZZ}) \otimes \mathcal{Q}(L^2(Y_c))$.
\end{lem}
\begin{proof}
First, note that $\Theta_{\rm pert}$ and $R_c$ are $C_0^\infty(Y_c)$ functions, so Prop.\ \ref{prop:ess.spec} applies. Namely, $\sigma_{\rm ess}(H_\Theta^{(Y_c)}(e^{ik}))\subset (2\NN+1)|b|$, so that it does not meet the support of $\varphi_n^2 - \varphi_n$. Therefore, each $(\varphi_n^2-\varphi_n)(H_\Theta^{(Y_c)}(e^{ik}))$ is at most finite-rank. Thus, each $\varpi'\big( \varphi_n \big( H_\Theta^{(Y_c)}(e^{ik}) \big)$ is a projection in $\mathcal{Q}(L^2(Y_c))$, with continuous dependence on $e^{ik}$ due to Corollary \ref{cor:continuous.loop}. Thus we obtain a projection
\begin{equation}
    \left\{ \varpi'\big( \varphi_n \big( H_\Theta^{(Y_c)}(e^{ik}) \big) \big) \right\}_{e^{ik} \in \widehat{\ZZ}} \in C(\widehat{\ZZ}) \otimes \mathcal{Q}(L^2(Y_c)). \label{eqn:family.Calkin}
\end{equation} 
This finishes the proof.
\end{proof}

\subsection{Invariant Roe algebras}
In the $\phi$-invariant setting, we only need to work in the $\ZZ$-invariant Roe algebra $C^*(X_c)^\ZZ$.
\begin{lem}\label{lem:Fourier.Roe}
We have isomorphisms
\begin{equation}
C^*(X_c)^\ZZ\cong C^*_r(\ZZ)\otimes C^*(Y_c)\cong C(\widehat{\ZZ})\otimes C^*(Y_c),\label{eqn:Roe.Fourier}
\end{equation}
where $C^*_r(\mathbb{Z})$ denotes the reduced group $C^*$-algebra for $\mathbb{Z}$.
\end{lem}

\begin{proof}
The second isomorphism in Eq.\ \eqref{eqn:Roe.Fourier} 
comes from the Fourier transform $C^*_r(\ZZ)\cong C(\widehat{\ZZ})$, while
the first isomorphism arises from the decomposition $L^2(X_c)\cong \ell^2(\ZZ)\otimes L^2(Y_c)$, as follows. With the fundamental domain $U_c$ as in Eq.\ \eqref{EqDefUc}, write $\Pi_n $ for the projection onto $L^2(U_c + 2n\pi/|c|)$. Since the subspaces $\Pi_nL^2(X_c)$ are identified by shift unitaries $S^n$, we get unitary isomorphisms 
\[ L^2(X_c) \cong \bigoplus_{n \in \mathbb{Z}} \Pi_nL^2(X_c) \cong \ell^2(\mathbb{Z}) \otimes L^2(U_c). \]
A $\mathbb{Z}$-invariant locally compact operator $T \in B(L^2(X_c))$ with finite propagation is decomposed into an infinite sum 
\[T = \sum_{n \in \mathbb{Z}}\Big( \sum_{m \in \mathbb{Z} } \Pi_{m+n} T \Pi_m\Big)  = \sum_{n\in \mathbb{Z}} S^n \otimes T_n \in B(\ell^2(\mathbb{Z}) \otimes L^2(U_c)),  \]
where $T_n:= S^{-n} \Pi_nT\Pi_0 \in B(L^2(U_c))$. 
This says that we have the tensor product decomposition $\mathbb{C}[X_c] \cong \mathbb{C}(\mathbb{Z}) \otimes \mathbb{C}[U_c]$ compatible with $L^2(X_c) \cong \ell^2(\mathbb{Z}) \otimes L^2(U_c)$, where $\mathbb{C}(\mathbb{Z}) \subset B(\ell^2(\mathbb{Z}))$ denotes the group algebra of $\mathbb{Z}$. By taking the $C^*$-algebra closure, we obtain
\[C^*(X_c)^{\mathbb{Z}} \cong C^*_r(\mathbb{Z}) \otimes C^*(U_c) \cong C^*_r(\mathbb{Z}) \otimes C^*(Y_c). \qedhere \] 
\end{proof}

The isomorphism Eq.\ \eqref{eqn:Roe.Fourier} also applies to the localized Roe algebras,
\begin{equation*}
C^*_{X_c}(\partial X_c^+)^\ZZ\cong C(\widehat{\ZZ})\otimes C^*_{Y_c}(\partial Y_c^+) = C(\widehat{\ZZ}) \otimes \mathcal{K}(L^2(Y_c)).
\end{equation*}
This says that $C^*(X_c)^\ZZ/C^*_{X_c}(\partial X_c^+)^\ZZ$ is a $C^*$-subalgebra of $C(\widehat{\ZZ}) \otimes \mathcal{Q}(L^2(Y_c))$ on which $\varpi'$ given in Eq.\ \eqref{eqn:family.Calkin} of Lemma \ref{lem:quotient.projection} coincides with $\varpi $ given in Eq.\ \eqref{eqn:quotient.SES}.

The isomorphism 
\[ 
\partial \colon K_0\left(C(\widehat{\ZZ}) \otimes \mathcal{Q}(L^2(Y_c))\right) \to K_1\left(C(\widehat{\ZZ}) \otimes \mathcal{K}(L^2(Y_c))\right) \cong \ZZ
\]
is given by the spectral flow of self-adjoint Fredholm operators \cite{Phillips}.
Hence the spectral gap filling at $\mu \in ((2n-1)|b|,(2n+1)|b|)$, proved in Theorem \ref{thm:helical.gap.filling}, is now understood as the spectral flow of catenoid Landau operators across $\mu$, and the number of eigenvalues crossing $\mu$ is measured by the $K_0$-class of the spectral projection $\sum_{2k+1 <\mu} \varpi' \big( \varphi_k\big(H_\Theta^{(X_c^+)} \big) \big)$. 
Eq.\ \eqref{eqn:nontrivial.connecting.map} shows that 
\[ \partial \bigg[ \sum_{2k+1 <\mu} \varpi' \Big( \varphi_k\big( H_\Theta^{(X_c^+)}\big) \Big) \bigg] = \pm n \in K_1\left(C(\widehat{\ZZ}) \otimes \mathcal{K}(L^2(Y_c))\right) \cong \ZZ,\]
since $\partial_{\rm MV} ({\rm Ind} (D))$ is a generator of $K_1\big(C^*_{X_c^+}(\partial X_c^+)^\ZZ\big)$.

\section*{Acknowledgements}
Y.K.\ acknowledges support from RIKEN iTHEMS and JSPS KAKENHI Grant Numbers 19K14544, JP17H06461 and JPMJCR19T2. 
M.L.\ acknowledges support from SFB 1085 ``Higher invariants'' of the DFG.
G.C.T.\ acknowledges support from Australian Research Council DP200100729.

\section*{Data Availability}
Data sharing not applicable to this article as no datasets were generated or analyzed during the current study.

\end{document}